\documentclass[11pt,a4paper,english]{article}
\usepackage[T1]{fontenc}
\usepackage[margin=1in]{geometry}
\usepackage{amsmath}
\usepackage{amsthm}
\usepackage{amstext}
\usepackage{amssymb}
\usepackage[unicode=true,pdfusetitle,
 bookmarks=true,bookmarksnumbered=false,bookmarksopen=false,
 breaklinks=false,pdfborder={0 0 1},backref=false,colorlinks=false]
 {hyperref}
 
{\begin{list}{}%
         {\setlength{\leftmargin}{#1}}%
         \item[]%
}
{\end{list}}

\theoremstyle{plain}
\newtheorem{theorem}{Theorem}

\newtheorem{lemma}[theorem]{Lemma}
\newtheorem{proposition}[theorem]{Proposition}
\newtheorem{claim}[theorem]{Claim}

\theoremstyle{remark}
\newtheorem*{remark}{Remark}

\theoremstyle{definition}
\newtheorem{definition}[theorem]{Definition}

\title{Set membership with non-adaptive bit probes}

\usepackage{authblk}

\author[1]{Mohit Garg\footnote{A part of this work was done when the first author was at the Tata Institute of Fundamental Research.}}
\author[2]{Jaikumar Radhakrishnan}
\affil[1]{Tokyo Institute of Technology, Tokyo,
	\texttt{garg.m.aa@m.titech.ac.jp}}
\affil[2]{Tata Institute of Fundamental Research, Mumbai,
	\texttt{jaikumar@tifr.res.in}}

\date{}

\newcommand{\lab}{\mathrm{lab}}

\newcommand{\pseudoedge}[1]{\stackrel{#1}{\longrightarrow}}
\newcommand{\upseudoedge}[1]{\stackrel{#1}{\longleftrightarrow}}
\newcommand{\floor}[1]{\left\lfloor #1 \right\rfloor}
\newcommand{\ceil}[1]{\left\lceil #1 \right\rceil}

\newcommand{\E}{\mathbb{E}}

\begin{document}
\maketitle
\begin{abstract}
We consider the non-adaptive bit-probe complexity of the set
membership problem, where a set $S$ of size at most $n$ from a
universe of size $m$ is to be represented as a short bit vector in
order to answer membership queries of the form ``Is $x$ in $S$?''  by
\textit{non-adaptively} probing the bit vector at $t$ places. Let
$s_N(m,n,t)$ be the minimum number of bits of storage needed for such
a scheme. Buhrman, Miltersen, Radhakrishnan, and
Srinivasan~\cite{BMRV2002} and Alon and Feige~\cite{AF2009}
investigated $s_N(m,n,t)$ for various ranges of the parameter $t$. We
show the following.

\emph{General upper bound ($t\geq5$ and odd):} For odd $t\geq 5$,
$s_N(m,n,t)=O(t m^{\frac{2}{t-1}} n^{1-\frac{2}{t-1}} \lg
\frac{2m}{n})$. This improves on a result of Buhrman {\em et al.}
that states for odd $t\geq5$, $s_N(m,n,t) = O(m^{\frac{4}{t+1}}n)$.
For small values of $t$ (odd $t \geq 3$ and $t \leq
  \frac{1}{10} \lg\lg m$) and $n \leq m^{1-\epsilon}$ ($\epsilon
  > 0$), we obtain adaptive schemes that use a little less space: $O(\exp(e^{2t})m^{\frac 2 {t+1}} n^{1 -
    \frac 2 {t+1}} \lg m)$.

    \emph{Three probes ($t=3$) lower bound:} We show that
$s_N(m,n,3)=\Omega(\sqrt{mn})$ for $n\geq n_0$ for some constant
$n_0$.  This improves on a result of Alon and Feige that states that
for $n\geq 16\lg m$, $s_N(m,n,3)=\Omega(\sqrt{\frac {mn}{\lg
    m}})$. The complexity of the non-adaptive scheme might, in
principle, depend on the function that is used to determine the answer
based on the three bits read (one may assume that all queries use the
same function). Let $s_N^f(m,n,3)$ be the minimum number of bits of
storage required in a three-probe non-adaptive scheme where the
function $f:\{0,1\}^3 \rightarrow \{0,1\}$ is used to answer the
queries. We show that for large class of functions $f$ (including the
majority function on three bits), we in fact have
$s_N(m,n,3)=\Omega(m^{1-\frac{1}{cn}})$ for $n\geq 4$ and some $c>0$.
In particular, three-probe non-adaptive schemes that use such query
functions $f$ do not give any asymptotic savings over the trivial
characteristic vector when $n \geq \log m$.
\end{abstract}

\section{Introduction}

The \emph{set membership problem} is a fundamental problem in the area
of data structures and information compression and retrieval. In its
abstract form we are given a subset $S$ of size at most $n$ from a universe of
size $m$ and required to represent it as a bit string so that
membership queries of the form ``Is $x$ in $S$?'' can be answered
using a small number of probes into the bit string. The characteristic
function representation provides a solution to this problem: just one
bit-probe is needed to answer queries, but all sets are represented
using $m$-bit strings (which is very wasteful when $n$ is promised to
be small).

The trade-off between the number of bits in the representation and the
number of probes is the subject of several previous works: it was
studied by Minsky and Papert in their 1969 book {\em Perceptrons}~\cite{MP1969}; more
recently, Buhrman, Miltersen, Radhakrishnan and Venkatesh~\cite{BMRV2002} 
showed the existence of randomized schemes that answer queries with just one bit probe and use near optimal space.
In contrast, they showed that deterministic schemes that answer queries by making a constant number of probes cannot use optimal space.
The deterministic worst-case trade-off for this problem was also considered in the same paper and in several
subsequent works (e.g., Radhakrishnan, Raman and Rao~\cite{RRR2001},
Alon and Feige~\cite{AF2009},
Radhakrishnan, Shah and Shannigrahi~\cite{RSS2010},
Viola~\cite{V2012},
Lewenstein, Munro, Nicholson and Raman~\cite{LMNR2014}, 
Garg and  Radhakrishnan~\cite{GR2015}).
For sets where each element is included with probability $p$,
Makhdoumi, Huang, M\'{e}dard and Polyanskiy~\cite{MHMP2015} showed, in
particular, that no savings over the characteristic vector can be
obtained in this case for non-adaptive schemes with t = 2.

In this work, we focus on deterministic schemes with \emph{non-adaptive} probes,
where the probes are made in parallel (or equivalently the location of
probes do not depend on the value read in previous probes). Such
schemes have been studied in several previous works. Let $s_N(m,n,t)$
be the minimum number of bits of storage required in order to answer
membership queries with $t$ non-adaptive probes. 

\begin{definition}
A non-adaptive $(m,n,s,t)$-scheme consists of a storage function and a
query scheme. The storage function has the form $\phi: {[m] \choose
\leq n} \rightarrow \{0,1\}^s$ that takes a set of size at most $n$
  and returns its $s$-bit representation. The query scheme associates
  with each element $x$ the $t$ probe locations $(i_1(x),\ldots,
  i_t(x)) \in [s]^t$ and a function $f_x: \{0,1\}^t \rightarrow
  \{0,1\}$. We require that for all $S \in {[m] \choose {\leq n}}$ and
  all $x \in [m]$: $x \in S$ iff
  $f_x(\phi(S)[i_1(x)],\phi(S)[i_2(x)],\ldots,\phi(S)[i_t(x)]) =
  1$. Let $s_N(m,n,t)$ denote the minimum $s$ such that there is an
  $(m,n,s,t)$-scheme. 
\end{definition}
In our discussion, we use $s(m,n,t)$ (without the subscript $N$) to denote the minimum space required for adaptive schemes. 
Using the above notation, we now describe
our results and their relation to what was known before. All asymptotic
claims below hold for large $m$.

\subsection{General non-adaptive schemes}\label{sec:generalResults}

\begin{theorem}[Result 1, non-adaptive schemes]
\label{thm:generalnonadaptive}
For odd $t\geq 5$, we have
\[ s_N(m,n,t)=O(t m^{\frac{2}{t-1}}
  n^{1-\frac{2}{t-1}} \lg \frac{2m}{n}).\]
\end{theorem}
In comparison, for odd $t\geq 5$, Buhrman {\em et al.} showed that
$s_N(m,n,t) = O(m^{\frac{4}{t+1}}n)$.  The exponent of $m$ in their
upper bound result is roughly four times the exponent of $m$ appearing
in their lower bound result.  Their schemes are non-adaptive and use
the MAJORITY function to answer membership queries. We exhibit schemes
that still use MAJORITY but need less space.  Buhrman {\em et al.}
also show a lower bound of $s(m,n,t) = \Omega(tm^{\frac{1}{t}}
n^{1-\frac{1}{t}})$ valid (even for adaptive schemes) when $n \leq
m^{1-\epsilon}$ (for $\epsilon > 0$ and $t \ll \lg m$). Note that the
exponent of $m$ in our result is twice the exponent of $m$
appearing in the lower bound result.
These schemes, as
well as the non-adaptive scheme for $t=4$ due to Alon and
Feige~\cite{AF2009}, have implications for the problem studied by
Makhdoumi {\em et al.}~\cite{MHMP2015}; unlike in the case of $t=2$,
significant savings are possible if $t\geq 4$, even with non-adaptive
schemes\footnote{We are grateful to Tom Courtade and Ashwin Pananjady
  for this observation.}.
  Using a similar proof idea, we obtain slightly better upper bound with adaptive schemes
when $t$ is small and $n$ is at most $m^{1-\epsilon}$.
\begin{theorem}[Result 2, adaptive schemes]
\label{thm:generaladaptive}
For odd $t \geq 3$ and $t \leq
  \frac{1}{10} \lg\lg m$ and for $n \leq m^{1-\epsilon}$ ($\epsilon
  > 0$), we have 
$ s(m,n,t)=O(\exp(e^{2t})m^{\frac 2 {t+1}} n^{1 - \frac 2 {t+1}} \lg m)$.
\end{theorem}

\emph{Technique.} To justify our claim, we need to describe the query
scheme, that is, $(i_1(x),i_2(x),$ $\ldots,i_t(x))$ for each $x \in
[m]$ and the query function $f_x:\{0,1\}^t \rightarrow \{0,1\}$.  For
$f_x$ we use the MAJORITY on $t$ bits ($t$ is odd). The locations to
be probed for each element will be obtained using a probabilistic
argument. Once a query scheme is fixed, we need to show how the
assignment to the memory is obtained. For this, we describe a
sequential algorithm.  We show that the random assignment of locations
ensures sufficient expansion allowing us to start with a greedy
argument arrange that most queries are answered correctly, and then
use Hall's bipartite graph matching theorem to find the required
assignment for the remaining elements. Versions of this argument have
been used in previous
works~\cite{LMSS2001,BMRV2002,GM2011,AF2009,GR2015}.


\subsection{Three non-adaptive probes}\label{sec:3NonAdaptiveProbes}
For one probe and $m\geq 2$, it is easy to show that no space can be
saved over the characteristic vector representation.  For two
non-adaptive probes, only for the special case $n=1$, some non-trivial
savings over the characteristic vector representation are possible:
$s_N(m,1,2)=\theta(\sqrt m))$.  For $n\geq 2$, Buhrman {\em et
  al.}~\cite{BMRV2002} showed $s_N(m,n,2)=m$.  The smallest number of
probes for which the complexity of problem with non-adaptive probes is
not settled is three. Observe that any scheme with two adaptive probes
can be converted to a scheme with three non-adaptive probes; the two
probe decision tree has at most three nodes. Thus, using the two
adaptive probes upper bound result of Garg and
Radhakrishnan~\cite{GR2015}, we have $s_N(m,n,3)\leq
s(m,n,2)=O(m^{1-\frac 1{4n+1}})$. Thus, non-trivial savings in space
over the characteristic vector representation is possible when $n=o(\lg
m)$. Consequently, the question is whether more space can be saved or
is this upper bound tight? We are not aware of any three-probe
non-adaptive scheme that manages with $o(m)$ space for sets of size
$\omega(\lg m)$. Alon and Feige~\cite{AF2009} show the following lower
bound: $s_N(m,n,3)=\Omega(\sqrt{\frac {mn}{\lg m}})$ for $n\geq 16\lg
m$.

In order to obtain better lower bounds for three-probe non-adaptive
schemes, we proceed as follows.  In any three-probe non-adaptive
scheme, the query scheme specifies, for each element, the three
locations to probe and a three variable boolean function to be applied
on three values read. In principle, for different elements, the query
scheme can specify different boolean functions. But since there are
only a finite number (256) of boolean functions on three variables,
some set of at least $m/256$ elements of the universe use a common
function. We may thus restrict attention to this part of the
universe, and assume that the function being employed to answer
queries is always the same.
Furthermore, we may place functions obtained from one another by
negating and permuting variables in a common equivalence class, and
restrict our attention to one representative in each class.  For three
variable boolean functions, P{\'o}lya counting yields that there are
twenty-two equivalence classes. This classification of the 256
functions into twenty-two classes is already available in the
literature~\cite{threeAry}. We show the following.

\begin{theorem}[Result 3]\label{result3}
\begin{enumerate}
\item[(a)] If the query function $f:\{0,1\}^3\rightarrow\{0,1\}$ is
  not equivalent to $(x,y,z)\mapsto (x\wedge y)\oplus z$ or
  $(x,y,z)\mapsto 1$ iff $x+y+z= 1$, then $s_N(m,n,3)=\Omega(m^{1-\frac
    1{cn}})$ for $n\geq 4$ and some $c>0$.
\item[(b)] If the query function $f:\{0,1\}^3\rightarrow\{0,1\}$ is equivalent to $(x,y,z)\mapsto (x\wedge y)\oplus z$ or $(x,y,z)\mapsto 1$
	iff $x+y+z=1$, then
	$s_N(m,n,3)=\Omega(\sqrt{mn})$.
\item[(c)]If the query function $f:\{0,1\}^3\rightarrow\{0,1\}$ is equivalent to $(x,y,z)\mapsto (x\wedge y)\oplus z$ and $\lg m\leq n\leq \frac{m}{\lg m}$, then $s_N(m,n,3)=\Omega(\sqrt{mn \frac{\lg \frac m n}{\lg \lg m}})$.
\end{enumerate}
\end{theorem}

The best upper bounds for non-adaptive schemes with four or more
probes use the MAJORITY function to answer membership queries. Our
result implies that for three non-adaptive probes, when queries are
answered by computing MAJORITY, the space required is at least
$\Omega(m^{1- \frac 1{cn}})$ for some constant $c$.  In fact, similar
lower bound holds if membership queries are answered using most
boolean functions. Our results do not yield a similar lower bound for
$(x,y,z)\mapsto (x\wedge y)\oplus z$ and $(x,y,z)\mapsto 1$ iff
$x+y+z=1$ types. For these two types of query functions, we get a
slightly better lower bound than what is implied by
\cite{AF2009}. Thus, further investigations on three probes
non-adaptive schemes need to focus on just $(x,y,z)\mapsto 1$ iff
$x+y+z=1$ and $(x,y,z)\mapsto (x\wedge y)\oplus z$ as the query
functions.

\emph{Technique.} As mentioned above, there are twenty-two types of functions for which we need to prove a lower bound.
Seven of the twenty-two classes contain functions that can be
represented by a decision tree of height at most two. Thus, for these
functions, the two probe adaptive lower bound in~\cite{GR2015} implies the
result. These functions are: constant 0, constant 1, the DICTATOR
function $(x,y,z) \mapsto x$, the function $(x,y,z)\mapsto x \wedge
y$, its complement $(x,y,z)\mapsto \bar{x} \vee \bar{y}$,
$(x,y,z)\mapsto (x\wedge y)\vee(\bar{x}\wedge z)$, and $(x,y,z)\mapsto
(x\wedge y)\vee (\bar{x}\wedge\bar{y})$.

After this, fifteen classes remain. Functions in some eleven of the
remaining fifteen classes admit a density argument, similar in spirit
to the adaptive two-probes lower-bound proof in~\cite{GR2015}.  To streamline the
argument, we classify these eleven classes into two parts. The first
part contains the MAJORITY function. The second part contains the AND
function, the ALL-EQUAL function, the functions
$(x,y,z)\mapsto (x\oplus y)\wedge z$, $(x,y,z)\mapsto (x\vee
y)\wedge z$, $(x,y,z)\mapsto (x\wedge y\wedge
z)\vee (\bar{y}\wedge\bar{z})$, and their complements. For functions
in the second part we deal with two functions---a function and its
complement---with a single proof. In these proofs, we produce sets $S$
and $T$ of size at most $n$ such that storing $S$ and not storing $T$
leads to a contradiction. The proof for the complement function works
with a small twist: storing $T$ and not storing $S$ leads to the
contradiction. Thus, these eleven cases are handled by six proofs. In
each of these proofs we roughly argue (sometimes probabilistically)
that if the scheme is valid, it must conceal a certain dense graph
that avoids small cycles. Standard graph theoretic results (the Moore
bound) that relate density and girth then gives us the lower bound.

For the remaining four classes, we employ linear-algebraic
arguments. Representatives chosen from these classes are PARITY,
$(x,y,z)\mapsto 1$ iff $x+y+z\neq 1$, $(x,y,z)\mapsto (x\wedge
y)\oplus z$, and $(x,y,z)\mapsto 1$ iff $x+y+z=1$. For PARITY and
$(x,y,z)\mapsto 1$ iff $x+y+z\neq 1$, we show using standard dimension
argument, that if the space used is smaller than the universe size
$m$, then there is some element $u\in[m]$ that is (linearly) dependent
on the other elements. Not storing the other elements, leaves the
scheme with no choice for $u$, thus leading to a contradiction. For
$(x,y,z)\mapsto(x\wedge y)\oplus z$ and $(x,y,z)\mapsto 1$ iff
$x+y+z=1$ a modification of an algebraic argument of Radhakrishnan,
Sen and Venkatesh~\cite{RSV2002} implies a lower bound of $\sqrt{mn}$.
(Interestingly, we need to choose an appropriate characteristic of the
field (2 or 3) based on which function we deal with.) For
$(x,y,z)\mapsto(x\wedge y)\oplus z$, we further improve on this
argument by employing random restrictions. These results together
improve the previous best lower bound (due to Alon and
Feige~\cite{AF2009}) irrespective of the query function used.


\section{General non-adaptive upper bound}
\label{sec:multi-probe-ub}

In this section, we prove the general non-adaptive upper bound result: Theorem~\ref{thm:generalnonadaptive}.

\begin{definition}
A non-adaptive $(m,s,t)$-graph is a bipartite graph $G$ with vertex
sets $U=[m]$ and $V$ ($|V|=ts$). $V$ is partitioned into $t$ disjoint
sets: $V_1,\ldots,V_t$; each $V_i$ has $s$ vertices. Every 
$u\in U$ has a unique neighbour in each $V_i$. A non-adaptive
$(m,s,t)$-graph naturally gives rise to a non-adaptive $(m, ts, t)$-query
scheme ${\cal T}_G$ as follows. We view the memory (an array $L$ 
of $ts$ bits) to be indexed by vertices in $V$. On receiving the query
``Is $u$ in $S$?'', we answer ``Yes'' iff the MAJORITY of the locations in the 
neighbourhood of $u$ contain a $1$. We say that the query scheme ${\cal T}_G$ 
is satisfiable for a set  $S\subseteq[m]$, if there is an assignment to the memory locations 
$(L[v] : v \in V)$, such that ${\cal T}_G$ correctly answers all 
queries of the form ``Is $x$ in $S$?''.
\end{definition}

We now restrict attention to odd $t\geq5$. First, we identify an
appropriate property of the underlying non-adaptive $(m,s,t)$-graph
$G$ that guarantees that ${\cal T}_G$ is satisfiable for all sets $S$
of size at most $n$.  We then show that such a graph exists for some
$s=O(m^{\frac{2}{t-1}} n^{1-\frac{2}{t-1}}\lg \frac {2m}{n})$.

\begin{definition}[Non-adaptive admissible graph] \label{def:t-admissible}
We say that a non-adaptive $(m,s,t)$-graph $G$ is admissible for sets of
size at most $n$ if the following two properties hold:
\begin{enumerate}
	\item[(P1)] $\forall R \subseteq [m]$ ($|R|\leq n + \ceil{2n \lg \frac {2m}{n}}$):
                $|\Gamma_{G}(R)| \geq \frac {t+1} {2} |R|$, where $\Gamma_G(R)$
                is the set of neighbors of $R$ in $G$.
        \item[(P2)] $\forall S \subseteq [m]$ ($|S| = n$): $|T_S| \leq \ceil{2n
                \lg \frac {2m}{n}}$, where $T_S = \{ y \in [m] \setminus S : |\Gamma_G(y)
                        \cap \Gamma_G(S)| \geq \frac {t+1}{2} \}$.
\end{enumerate}
\end{definition}
Our theorem will follow from the following claims.
\begin{lemma} \label{lm:t-admissibleImpliesSatisfiable}
If a non-adaptive $(m,s,t)$-graph $G$ is admissible for sets of size
at most $n$, then the non-adaptive $(m, ts, t)$-query scheme ${\cal
  T}_{G}$ is satisfiable for every set $S$ of size at most $n$.
\end{lemma}
\begin{lemma} \label{lm:t-admissibleExists}
There is a non-adaptive $(m,s,t)$-graph, with $s=O(m^{\frac{2}{t-1}}
n^{1-\frac{2}{t-1}}\lg \frac {2m}{n})$, that is admissible for every set
$S\subseteq [m]$ of size at most $n$. 
\end{lemma}

\begin{proof}[Proof of Lemma~\ref{lm:t-admissibleImpliesSatisfiable}]
Fix an admissible graph $G$. 
Thus, $G$ satisfies (P1) and (P2)
above. Fix a set $S \subseteq [m]$ of size at most $n$.  We will show
that there is a 0-1 assignment to the memory such that all queries are
answered correctly by ${\cal T}_G$.

Let $S' \subseteq [m]$ be such that $S \subseteq S'$ and $|S'| =
n$. From (P2), we know $|T_{S'}| \leq \ceil{2n \lg \frac {2m}{n}}$. Hence, $|S'
\cup T_{S'}| \leq n + \ceil{2n \lg \frac {2m}{n}}$. From (P1) and Hall's theorem,
we may assign to each element $u \in S' \cup T_{S'}$ a set $A_u
\subseteq V$ such that (i) $|A_u| = \frac {t+1}{2}$ and (ii) the
$A_u$'s are disjoint.  For each $u\in S \subseteq S'$, we assign the
value 1 to all locations in $A_u$.  For each $u\in (S' \cup T_{S'})
\setminus S$, we assign the value 0 to all locations in $A_u$. Since
$\frac {t+1}{2} > \frac t 2$, all queries for $u\in S' \cup T_{S'}$ are
answered correctly.

Assign 0 to all locations in $\Gamma_G([m]
\setminus (S' \cup T_{S'}))$. For $y\in [m]
\setminus (S' \cup T_{S'})$, $|\Gamma_G(y) \cap \Gamma_G(S)|\leq \frac {t-1} {2}$.
As a result, queries for elements in $[m] \setminus 
(S' \cup T_{S'})$ are answered correctly, as the majority evaluates to 0 for each 
one of them. 
\end{proof}

\begin{proof}[Proof of Lemma~\ref{lm:t-admissibleExists}]
In the following, set 
\[ s = \ceil{60 m^{\frac{2}{t-1}} n^{1-\frac{2}{t-1}}\lg \frac {2m}{n}}.\]
We show that a suitable random non-adaptive $(m,s,t)$-graph $G$ is
admissible for sets of size at most $n$ with positive probability. The
graph $G$ is constructed as follows. Recall that $V = \bigcup_i
V_i$. For each $u \in U$, one neighbor is chosen uniformly and
independently in each $V_i$.
\begin{description}
\item[(P1) holds.] If (P1) fails, then for some non-empty $W \subseteq
  U$, $(|W| \leq n + \ceil{2n \lg \frac {2m}{n}})$, we have $|\Gamma_G(W)| \leq
  \frac {t+1}{2}|W| - 1$. Fix a set $W$ of size $r \geq 1$ and $L
  \subseteq V$ of size $\frac{t+1}{2}r - 1$. Let $L$ have
  $\ell_i$ elements in $V_i$; thus, $\sum_i \ell_i =
  \frac{t+1}{2}r - 1$. Then,
\[
\Pr[\Gamma_G(W) \subseteq L] \leq \prod_{i=1}^t
\left(\frac{\ell_i}{|V_i|}\right)^{r} \leq
\left(\frac{(\frac{t+1}{2})r-1}{ts}\right)^{tr},
\]
where the last inequality is a consequence of GM $\leq$ AM. We conclude, using the union bound over choices of $W$ and $L$, that (P1) fails with probability 
at most
\begin{align}
& \sum_{r=1}^{n+\ceil{2n\lg \frac {2m}{n}}} {m\choose r} 
{{ts} \choose \frac{t+1}{2}r-1}\left(\frac{\frac{t+1}{2}r-1}{ts}\right)^{tr}\\
&\leq 
  \sum_{r=1}^{n+\ceil{2n\lg \frac {2m}{n}}}
  \left(\frac{em}{r}\right)^r\left(\frac{tes}{\frac{t+1}{2}r-1}\right)^
  {\frac{t+1}{2}r-1}  
 \left(\frac{\frac{t+1}{2}r-1}{ts}\right)^{tr} \nonumber \\
&\leq
\sum_{r=1}^{n+\ceil{2n\lg\frac {2m}{n}}}
\left[\frac{(e^{\frac{t+3}2 - \frac 1 r}) m  r^{\frac{t-1}2-1 + \frac{1}{r}} }{ 
(s^{\frac 1 r})s^{\frac{t-1}2}} \right]^r \leq \frac{1}{3}, 
\end{align}
where the last inequality holds because we have chosen $s$ large enough.

\item[(P2) holds.] For (P2) to fail, there must exist a set $S\subseteq [m]$ 
of size $n$ such that $|T_S| > \ceil{2n \lg \frac {2m}{n}}$. Fix a set $S$ of size $n$.
Fix a $y \in [m] \setminus S$.
\[
\Pr[y \in T_S]\leq {t \choose \frac {t+1}2}\left(\frac n s\right)^{\frac {t+1}2}
\leq \frac {n} {10m},
\]
where the last inequality holds because of choice of $s$ and $m$ is
large. Thus, $\E[|T_S|]  \leq \frac {n}{10}$.
To conclude that $|T_S|$ is bounded with high probability, we will use the following version of Chernoff bound: if $X = \sum_{i=1}^N X_i$, where each random variable $X_i \in \{0,1\}$ independently, then if $\gamma > 2e \E[X]$, then  $\Pr[X > \gamma] \leq 2^{-\gamma}$. Then,  for all large $m$, 
\[ \Pr[ |T_S| > 2n \lg \frac {2m}{n}] \leq 2^{- 2n \lg \frac {2m}{n}}. \]
Using the union bound, we conclude that
\[\Pr[\mbox{(P2) fails}]
\leq \left( \frac {em} n \right)^n 2^{-2n \lg \frac {2m}{n}}
\leq \frac 1 3.
\]
Thus, with probability at least $\frac 1 3$ the random graph $G$ is admissible.
\end{description}
\end{proof}


\section{General upper bound: adaptive}
\label{sec:multi-probe-ub-adaptive}
In this section, we prove the general adaptive upper bound result: Theorem~\ref{thm:generaladaptive}.
We will use the following definitions from previous works.

\begin{definition}[storing scheme, query scheme, scheme, systematic]
An $(m,n,s)$-{\em storing scheme} is a method for representing a
subset of size at most $n$ of a universe of size $m$ as an $s$-bit
string.  Formally, an $(m,n,s)$-storing scheme is a map $\phi$ from
${{[m]} \choose {\leq n}}$ to $\{0,1\}^s$.  

A deterministic
$(m,s,t)$-{\em query scheme} is a family $\{T_u\}_{u\in [m]}$ of $m$
Boolean decision trees of depth at most $t$. Each internal node in a
decision tree is marked with an index between $1$ and $s$, indicating
the address of a bit in an $s$-bit data structure. For each internal
node, there is one outgoing edge labeled ``0'' and one labeled ``1''.
The leaf nodes of every tree are marked `Yes' or `No'. Such a tree
$T_u$ induces a map from $\{0,1\}^s$ to \{Yes, No\}; this map will
also be referred to as $T_u$.  

An $(m,n,s)$-storing scheme $\phi$ and an $(m,s,t)$-query scheme
$\{T_u\}_{u \in [m]}$ together form an {\em $(m,n,s,t)$-scheme} if
$\forall S\in {[m] \choose \leq n},\, \forall u\in [m]: T_u(\phi(S))
=$Yes if and only if $u \in S$.  Let $s(m,n,t)$ be the minimum $s$
such that there is an $(m,n,s,t)$-scheme.

We say that an $(m,n,s,t)$-scheme is {\em systematic} if the value
returned by each of its trees $T_u$ is equal to the last bit it reads
(interpreting $0$ as No/False and $1$ as Yes/True).
\end{definition}
\newcommand{\leaves}{\mathsf{leaves}}

In order to show that $s(m,n,t)$ is small, we will exhibit efficient
adaptive schemes to store sets of size {\em exactly}
$n$. This will
imply our bound (where we allow sets of size {\em at most} $n$)
because we may {\em pad} the universe with $n$ additional elements,
and extend $S$ ($|S| \leq n$)by adding $n-|S|$ additional elements, to
get a subset is of size exactly $n$ in a universe of size $m+n \leq
2m$.

\begin{definition} 
An adaptive $(m,s,t)$-graph is a bipartite graph $G$ with vertex sets $U=[m]$
and $V$ ($|V|=(2^t-1)s$). $V$ is partitioned into $2^t-1$ disjoint sets:
$A$, $A_0$, $A_1$, $A_{00}$,\ldots, that is, one $A_{\sigma}$ for each $\sigma
\in \{0,1\}^{\leq (t-1)}$; each $A_{\sigma}$ has $s$ vertices. Between
each $u \in U$ and each $A_{\sigma}$ there is exactly one edge. 
Let $V_i:=\cup_{\sigma: |\sigma|=i-1} A_\sigma$.
An $(m,s,t)$-graph naturally gives rise to a systematic
$(m,(2^t-1)s, t)$-query scheme ${\cal T}_G$ as follows. We view the
memory (an array $L$ of $(2^t-1)s$ bits) to be indexed by vertices in
$V$. For query element $u\in U$, if the first ${i-1}$ probes resulted
in values $\sigma \in \{0,1\}^{i-1}$, then the $i$-th probe is made to
the location indexed by the unique neighbor of $u$ in $A_{\sigma}$. In
particular, the $i$-th probe is made at a location in $V_i$. 
We answer ``Yes'' iff the last bit read is $1$.
In addition, we use  following notation.
We refer to $V_t$ as the leaves of $G$ and for $y \in [m]$, let  
$\leaves(y):=V_t \cap \Gamma_G(y)$. For $R \subseteq [m]$, let
$\leaves(R):=V_t \cap \Gamma_G(R)$.

We say that the query scheme ${\cal T}_{G}$ is satisfiable for a set
$S \subseteq [m]$, if there is an assignment to the memory locations
$(L[v]: v \in V)$, such that ${\cal T}_{G}$ correctly answers all
queries of the form ``Is $x$ in $S$?''.
\end{definition}

We assume that $t \geq 3$ is odd and show that $\forall \epsilon >0$ 
$\forall n \leq m^{1-\epsilon}$ $\forall t \leq \frac{1}{10}\lg\lg m$
$s(m,n,t)=O(\exp(e^{2t})m^{\frac 2 {t+1}} n^{1 - \frac 2 {t+1}} \lg m)$. 
Our $t$-probe scheme will
have two parts: a $t_1$-probe non-adaptive part and a $t_2$-probe
adaptive part, such that $t_1 + t_2 = t$. The respective parts will be
based on appropriate non-adaptive $(m,s,t_1)$-graph $G_1$ and adaptive
$(m,s,t_2)$-graph $G_2$ respectively. To decide set membership, we
check set membership in the two parts separately and take the AND,
that is, we answer ``Yes'' iff all bits read in ${\cal T}_{G_1}$ are $1$
and the last bit read in ${\cal T}_{G_2}$ is $1$.  We refer to this
scheme as ${\cal T}_{G_1} \wedge {\cal T}_{G_2}$.

First, we identify appropriate properties of the underlying graphs
$G_1$ and $G_2$ that guarantee that all queries are answered correctly
for sets of size $n$.  We then show that such graphs exist with
$s=O(\exp(e^{2t}-t)m^{\frac 2 {t+1}} n^{1 - \frac 2 {t+1}} \lg m)$.

\newcommand{\surv}{\mathsf{survivors}} We will use the following
constants in our calculations: $\alpha:= 2^{t_2}-1$ and $\beta:=
2^{t_2} - t_2$.  Note that $\alpha$ is the total number of nodes in a
$t_2$-probe adaptive decision tree. In any such decision tree, for
every choice of $\beta$ nodes and every choice $b \in \{0,1\}$ of the
answer, it is possible to assign values to those $\beta$ nodes so that
the decision tree returns the answer $b$.

\begin{definition}[admissible-pair] \label{def:admissible-pair}
We say that a non-adaptive $(m,s,t_1)$-graph $G_1$ and an adaptive
$(m,s,t_2)$-graph $G_2$ form an admissible pair $(G_1,G_2)$ for
sets of size $n$ if the following conditions hold.
\begin{enumerate}
\item[(P1)]
 $\forall S \subseteq [m]$ ($|S| = n$):
$|{\surv}(S)| \leq 10 m \left( \frac n s \right)^{t_1}$,
where ${\surv}(S) = \{y \notin S: \Gamma_{G_1}({y})\subseteq \Gamma_{G_1}(S)\}$.
\item[(P2)] For $S \subseteq [m]$ ($|S| = n$), let 
${\surv}^+(S) = \{y \in {\surv}(S): 
{\leaves_{G_2}}(S) \cap {\leaves_{G_2}}(y) \neq \emptyset \}$. 
Then, 
$\forall S \subseteq [m]$ ($|S| = n$) 
$\forall T \subseteq S \cup {\surv}^+(S)$:
$\Gamma_{G_2}(T) \geq \beta |T|$.
\end{enumerate}
\end{definition}


\begin{lemma} \label{lm:t1t2-admissibleImpliesSatisfiable}
If a non-adaptive $(m,s,t_1)$-graph $G_1$ and an adaptive $(m,s,t_2)$-graph 
$G_2$ form an admissible pair for sets of size $n$, then 
the query scheme ${\cal T}_{G_1} \wedge {\cal T}_{G_2}$ is satisfiable 
for every set $S\subseteq [m]$ of size $n$.
\end{lemma}

\begin{lemma} \label{lm:t1t2-admissibleExists}
Let $t \geq 3$ be an odd number; let $t_1 = \frac{t-3}{2}$ and $t_2=\frac{t+3}{2}$.  
Then, there exist an admissible pair of graphs consisting of a non-adaptive $(m,s,t_1)$-graph $G_1$ and an adaptive $(m,s,t_2)$-graph $G_2$ with $s=O(\exp(e^{2t}-t)m^{\frac 2 {t+1}} n^{1 - \frac 2 {t+1}} \lg m)$.
\end{lemma}

\begin{proof}[Proof of Lemma~\ref{lm:t1t2-admissibleImpliesSatisfiable}]
Fix an admissible pair $(G_1, G_2)$. 
Thus, $G_1$ satisfies (P1) and $G_2$ satisfies (P2) above. Fix a set
$S \subseteq [m]$ of size $n$. We will show that there is an assignment such that
${\cal T}_{G_1} \wedge {\cal T}_{G_2}$ answers all questions of the form ``Is $x$ in $S$?'' correctly.

The assignment is constructed as follows. Assign 1 to all locations in $\Gamma_{G_1}(S)$ and 0 to the remaining locations in $\Gamma_{G_1}([m])$. Thus, ${\cal T}_{G_1}$ answers ``Yes'' for all query elements in $S$ and answers ``No'' for all query elements outside $S \cup \surv(S)$. However, it (incorrectly) answers ``Yes'' for elements in $\surv(S)$. We will now argue that these {\em false positives} can be eliminated using the scheme ${\cal T}_{G_2}$.

Using (P2) and Hall's theorem, we may assign to each element $u \in S
\cup {\surv}^+(S)$ a set $L_u \subseteq V(G_2)$ such that (i)$|L_u| =
\beta$ and (ii) the $L_u$'s are disjoint. Set $b_u = 1$ for $u\in S$ and
$b_u = 0$ for $u \in \surv^+(S)$ (some of the false positives).  As observed
above for each $u \in S \cup \surv^+(S)$ we may set the values in the
locations in $L_u$ such that the value returned on the query element
$u$ is precisely $b_u$.  Since the $L_u$'s are disjoint we may take such
an action independently for each $u$.  After this partial assignment,
it remains to ensure that queries for elements $y \in
\surv(S)\setminus \surv^+(S)$ (the remaining false positives) return a ``No''. Consider any such 
$y$. By the definition of $\surv^+(S)$, no location in
$\leaves_{G_2}(y)$ has been assigned a value in the above partial
assignment. Now, assign 0 to all unassigned locations in $V(G_2)$. Thus
${\cal T}_{G_2}$ returns the answer ``No'' for queries from
${\surv}(S) \setminus {\surv}^+(S)$.
\end{proof}

\begin{proof}[Proof of Lemma~\ref{lm:t1t2-admissibleExists}] In the following, let 
\[ s = \ceil{\exp(e^{2t}-t) m^{\frac{2}{t+1}} n^{1 - \frac{2}{t+1}}\lg m}.\]
We will construct the non-adaptive $(m,s,t_1)$-graph $G_1$ and the 
$(m,s,t_2)$-graph $G_2$ randomly, and show that with positive probability the pair $(G_1,G_2)$ is admissible. 
The graph $G_1$ is constructed as in the proof of 
Lemma~\ref{lm:t-admissibleExists}, and the analysis is similar. Recall that $V(G_1) = \bigcup_{i \in [t_1]} V_i(G_1)$.
For each $u \in U$, one neighbor is chosen uniformly and independently from each $V_i(G_1)$.

\begin{description}
\item[(P1) holds.] Fix a set $S$ of size $n$. Then, 
$\E[|{\surv}(S)|] \leq (m-n)\left(\frac n s\right)^{t_1} 
\leq m \left( \frac n s \right)^{t_1}$.
As before, using the Chernoff bound, we conclude that
\[ \Pr[|{\surv}(S)| > 10 m \left( \frac n s \right)^{t_1}] \leq 
2^{-10m\left(\frac n s  \right)^{t_1}}.\]
Then, by the union bound,
\begin{eqnarray*}
	\Pr[\mbox{P1 fails}]&\leq& {m \choose n}2^{-10m\left(\frac n s  \right)^{t_1}} \\
&\leq& \frac 1 {10},
\end{eqnarray*}
where the last inequality follows from our choice of $s$.

Fix a graph $G_1$ such that (P1) holds. The random graph $G_2$ is
constructed as follows. Recall that $V({G_2})= \bigcup_{z \in
  \{0,1\}^{\leq t_2 - 1}} A_z$.  For each $u \in [m]$, one neighbor is
chosen uniformly and independently from each $A_z$.
\item[(P2) holds]
To establish (P2), we need to show that all sets of the form $S' \cup
R$, where $S' \subseteq S$ and $R \subseteq \surv^+(S)$ expand. To
restrict the choices for $R$, we first show in
Claim~\ref{cl:survivorplus} (a) that with high probability
$\surv^+(S)$ is small. Then, using direct calculations, we show that
whp the required expansion is available in the random graph $G_2$.
\newcommand{\event}{\mathcal{E}}
\begin{claim} \label{cl:survivorplus} 
\begin{enumerate}
\item[(a)] Let $\event_a \equiv \forall S \subseteq [m] (|S| = n):
  |{\surv}^+(S)| \leq 100 \cdot 2^{t_2} m\left(\frac n s \right)^{t_1 + 1}$;
  then, $\Pr[\event_a] \geq \frac{9}{10}$.
\item[(b)] Let $\event_b \equiv \forall R \subseteq [m]\, (|R| \leq n
  + \ceil{n \lg m}): |\Gamma_{G_2}(R)| \geq \beta|R|$; then,
  $\Pr[\event_b] \geq \frac{9}{10}$.
\item[(c)] Let $\event_c = \forall S \subseteq [m] (|S| = n), \forall
  S' \subseteq S, \forall R \subseteq {\surv}^+(S) (\ceil{n \lg m} \leq
  |R| \leq 100 \cdot 2^{t_2} m\left(\frac n s \right)^{t_1 + 1}):
  |\Gamma_{G_2}(S' \cup R)| \geq \beta |S' \cup R|$; then,
  $\Pr[\event_c] \geq \frac{9}{10}$.
\end{enumerate}
\end{claim}
\item[Proof of claim~\ref{cl:survivorplus}.]  Part (a) follows by a routine application of Chernoff bound, as in several previous proofs. For a set $S$ of size $n$, we have $\E[{\surv}^+(S)]\leq |\surv(S)|2^{t_2}(\frac n s)$
        $\leq 2^{t_2}10m\left( \frac n s \right)^{t_1 + 1}$.  Then,
\begin{eqnarray*}
\Pr[\neg \event_a]&\leq& {m \choose n}2^{-2^{t_2}10m\left(\frac n s
\right)^{t_1+1}} \\ 
&\leq& \frac 1 {10},
\end{eqnarray*}
where the last inequality holds because of our choice of $s$.

Next consider part (b).  If $\event_b$ does not hold, then for some
non-empty $W \subseteq [m]$, $(|W| \leq n + \ceil{n \lg m})$, we have
$|\Gamma_{G_2}(W)| \leq \beta |W| -1$. Fix a set $W$ of size $r \geq
1$ and $L \subseteq V(G_2)$ of size $\beta r-1$. Let $L$ have
$\ell_z$ elements in $A_z$. Then,
\[
\Pr[\Gamma_{G_2}(W) \subseteq L] \leq \prod_z \left(\frac{\ell_z}{|A_z|}\right)^{r} \leq
\left(\frac{\beta r-1}{\alpha s}\right)^{ \alpha r}.
\]
We conclude, using the union bound over choices of $W$ and $L$, that the probability that 
$\event_b$ does not hold is at most
\begin{eqnarray*}
&& \sum_{r=1}^{n+\ceil{n\lg m}} {m\choose r}
{{\alpha s} \choose \beta r-1}\left( \frac{\beta r-1}{\alpha s} \right)^{\alpha r}\\
&\leq&
\sum_{r=1}^{n+\ceil{n\lg m}}
\left(\frac {em}{r}\right)^r \left( \frac{\alpha es}{\beta r-1} \right)^{\beta r-1}
\left(\frac{\beta r-1}{\alpha s}\right)^{\alpha r}\\
&\leq&
\sum_{r=1}^{n+\ceil{n\lg m}}
\left(\frac{\beta r}{\alpha es}\right) \left[ \frac {em}{r} e^\beta \left(\frac {\beta r}{\alpha s}\right)^{\alpha-\beta} \right]^r\\
&\leq&
\sum_{r=1}^{n+\ceil{n\lg m}}
\left(\frac{\beta r}{\alpha es}\right) \left[ e^{\beta +1}\left( \frac {\beta}{\alpha} \right)^{\alpha-\beta}  \left(\frac {mr^{\alpha-\beta-1}}{s^{\alpha-\beta}}\right) \right]^r\\
&\leq& \frac {1}{10},
\end{eqnarray*}
where the last inequality holds because of our choice of $s$.

Finally, we justify part (c).  To bound the probability that
$\event_c$ fails, we consider a set $S \subseteq [m]$ of size $n$, a
subset $S' \subseteq S$ of size $i$ (say), a subset $R \subseteq {\surv}^+(S)$ of
size $r$ (where $\ceil{n\lg m} \leq r \leq 100 \cdot 2^{t_2} m\left(\frac n
s \right)^{t_1 + 1}$) and $L \subseteq V(G_2)$ of size $ \ell = 
\beta (i+r) $ and define the event
\[ \event(S,S',R,L) \equiv (\forall y \in R: \leaves_{G_2}(S) \cap \leaves_{G_2}(y) \neq \emptyset) \wedge \Gamma_{G_2} (S' \cup R) \subseteq L.\]
Then,
\begin{align}
 \Pr[\event(S,S',R,L)] &\leq  \left( \frac {2^{t_2}n}{s} \right)^r 
\left( \frac {\ell}{(\alpha-1) s} \right)^{(\alpha-1)r} 
\left( \frac {\ell}{\alpha s} \right)^{\alpha i}\\
&\leq  \left( \frac {2^{t_2}n}{s} \right)^r
      \left( \frac {\beta(i+r)}{(\alpha-1) s} \right)^{(\alpha-1)(i+r)} 
     \left( \frac{\beta(i+r)}{\alpha s} \right)^{i},
\end{align} 
where the factor $\left( \frac {2^{t_2}n}{s} \right)^r $ is justified
because of the requirement that every $y \in R$ has at least one
neighbour in $\leaves_{G_2}(S)$; the factor $\left( \frac
{\ell}{(\alpha-1) s} \right)^{(\alpha-1)r}$ is justified because all
the remaining neighbours must lie in $L$ (we use AM $\geq$ GM); the
last factor $\left( \frac {\ell}{\alpha s} \right)^{\alpha i}$ is
justified because all neighbors of elements in $S'$ lie in $L$ (again
we use AM $\geq$ GM). To complete the argument we apply the union
bound over the choices of $(S,S',R,L)$. Note that we may restrict
attention to $\ell = \beta (i+r)$ (because for our choice of $s$, we
have $\beta(i+r) \leq |V(G_2)| = \alpha s$).  Thus, the probability
that $\event_c$ fails to hold is at most
     \[ \sum_{S,S',R,L} \Pr[\event_c(S,S',R,L)],\] 
     where $S$ ranges over sets of size $n$, $S' \subseteq S$ of size $i$, $R
     \subseteq \surv(S)$ of size $r$ such that $\floor{n \lg m} \leq
     r \leq 100 2^{t_2}m\left( \frac {n}{s} \right)^{t_1 + 1}$, $L$ is
     a subset of $V(G_2)$ of size $\beta(i+r)$. We evaluate this sum
     as follows.
\begin{align}
 & \sum_{r} \sum_{i} {m \choose n}
        {\floor{ 10m \left( \frac {n}{s} \right)^{t_1}} \choose r}
        {n \choose i}
        {\alpha s \choose \beta (i+r)}
        \left( \frac {2^{t_2}n}{s} \right)^{r}
        \left(\frac {\beta(i+r)}{(\alpha-1) s} \right)^{(\alpha-1)(i+r)} 
        \left(\frac{\beta(i+r)}{\alpha s} \right)^{i}
        \\
&\leq
  \sum_{r} \sum_{i} 
   \left[ \left(\frac {em}{n} \right)^{\frac {n}{i+r}} 
                \left( \frac {10 e m \left( \frac {n}{s} \right)^{t_1}} {r} \right)^{\frac {r}{i+r}}
		{n \choose i}^{\frac {1}{1+r}}
                \left( \frac {\beta(i+r)}{(\alpha-1)s} \right)^{\alpha-1} \right.
\nonumber\\
&
\left.
        \left(\frac{e\alpha s}{\beta(i+r)}\right)^\beta
        \left( \frac {2^{t_2}n}{s} \right)^{\frac{r}{i+r}}
        \left(\frac {\beta(i+r)}{(\alpha-1) s} \right)^{\frac{i}{i+r}} 
\right]^{i+r}\\
&\leq
  \sum_{r} \sum_{i} 
   \left[ \left(\frac {em}{n} \right)^{\frac {n}{i+r}} 
                \left( \frac {10 e m \left( \frac {n}{s} \right)^{t_1}} {r} \right)^{\frac {r}{i+r}}
                {n \choose i}^{\frac {1}{1+r}}
		\left( \frac {\beta(i+r)}{(\alpha-1)s} \right)^{\alpha-\beta-1} \right.
\nonumber\\
&
\left.
                \left(\frac{e\alpha}{\alpha-1} \right)^\beta
                \left( \frac {\beta(i+r)}{2^{t_2}n (\alpha-1)} \right)^{ \frac {i}{i+r} } 
                \left( \frac {2^{t_2}n}{s} \right)   
\right]^{i+r}.
\end{align}
We will show that the quantity inside the square brackets is at most
$\frac{1}{2}$. Then, since $r \geq n\lg m$ and $i \geq 0$
\[ \Pr[\neg \event_c] \leq \left(\sum_r 2^{-r}\right) \left(\sum_i 2^{-i} \right) \leq \frac{1}{10}.\]
The quantity in the brackets can be decomposed as a product of two
factors, which we will bound separately.
\begin{description}
\item[Factor 1:] Consider the following contributions
\[\left(\frac {em}{n} \right)^{\frac {n}{i+r}} 
  (10 e)^{\frac{r}{i+r}} 
  {n \choose i}^{\frac {1}{i+r}}
  \left(\frac{e\alpha}{\alpha-1}\right)^{\beta}
  \left(\frac {\beta(i+r)}{2^{t_2}n(\alpha-1)} \right)^{\frac{i}{i+r}}.
\]
Since $r \geq n \lg m$ and $i \leq n$, we have $\frac{i}{i+r} \leq
\frac{n}{n+r} \leq \frac{1}{\lg m} \leq \frac{1}{\lg_e m}$. Thus, for all large enough $m$, 
this quantity is at most
\[ e^2 \cdot 10e \cdot e^2\cdot (2e)^{\beta} \cdot e \leq \exp(e^{2t}-t).\]

\item[Factor 2:] We next bound the contribution for the remaining factors.
\begin{align}
& \left( \frac{m (\frac{n}{s})^{t_1}}{r}\right)^{\frac{r}{i+r}}
\left(\frac{\beta(i+r)}{(\alpha -1) s}\right)^{\alpha-\beta -1}\left(\frac{2^{t_2} n}{s}\right)\\
&\leq \left( \frac{m (\frac{n}{s})^{t_1}}{r}\right)
\left(\frac{2r}{ s}\right)^{\alpha-\beta -1}
\left(\frac{2^{t_2} n}{s}\right) \label{eq:dropexponent} \\
&= \frac{mn^{t_1+1}2^{\alpha-\beta+t_2 -1} r^{\alpha-\beta -2}}{s^{\alpha-\beta +t_1}}.
\end{align}
To justify (\ref{eq:dropexponent}),
recall that $r \leq 100 \cdot 2^{t_2} m\left(\frac n s \right)^{t_1 + 1}$
and \newline
$s = \ceil{\exp(e^{2t}-t) m^{\frac{2}{t+1}} n^{1 - \frac{2}{t+1}}\lg m}$; thus $\frac{m (\frac{n}{s})^{t_1}}{r} \geq 1$.
Then, the above quantity is bounded by
\begin{align}
& \frac{mn^{t_1+1}2^{2(t_2-1)} \left(100 \cdot 2^{t_2} m n^{t_1+1}\right)^{\alpha-\beta-2}}{s^{(t_1+1)(\alpha-\beta - 2)} s^{\alpha - \beta + t_1}}\\
&\leq \left(\frac{100\cdot 2^{2t_2} m n^{t_1+1}}{s^{t_1+2}}\right)^{\alpha-\beta-1}.
\end{align}
\end{description}
Thus, since $s = \ceil{\exp(e^{2t}-t) m^{\frac{2}{t+1}} n^{1 -
    \frac{2}{t+1}}\lg m}$, then the product of the factors is at most
$\frac{1}{10}$, as required.
\end{description}
\end{proof}


\section{Three non-adaptive probes lower bound}

In this section, we prove the three probe lower bound result: Theorem~\ref{result3}. 

\begin{definition}[Equivalent]
	Two boolean functions are called equivalent if one can be obtained from the other by negating and permuting the variables.
\end{definition}

\begin{proposition}
	Let $f,g:\{0,1\}^t \rightarrow \{0,1\}$ be equivalent. If $s_1$ and $s_2$ are the minimum bits of space required for 
	non-adaptive $(m,n,s_1,t)$ and $(m,n,s_2,t)$-schemes with query functions $f$ and $g$ respectively, then $s_1=s_2$.
\end{proposition}

For three variable boolean functions, there are twenty-two equivalence classes (see~\cite{threeAry}). To prove Theorem~\ref{result3},
we provide proofs for these twenty-two query functions, each from a different class. In many proofs below we assume that the memory consists of three arrays of size $s$ each, and the three probes
are made on different arrays. Given any scheme that uses space $s$, we can always modify it to meet our assumption, by expanding the space
by factor 3.

\subsection{Decision trees of height two}

Seven of the twenty-two classes contain functions that can be
represented by a decision tree of height at most two. Thus, for these
functions, the two probe adaptive lower bound~\cite{GR2015} implies the
result. These functions are: constant 0, constant 1, the DICTATOR
function $(x,y,z) \mapsto x$, the function $(x,y,z)\mapsto x \wedge
y$, its complement $(x,y,z)\mapsto \bar{x} \vee \bar{y}$,
$(x,y,z)\mapsto (x\wedge y)\vee(\bar{x}\wedge z)$, and $(x,y,z)\mapsto
(x\wedge y)\vee (\bar{x}\wedge\bar{y})$.

\newcommand{\Maj}{\mathsf{Maj}}

\subsection{MAJORITY}
Let $\Phi$ be a non-adaptive $(m,n,s,3)$-scheme with MAJORITY as the query function.
The memory is a bit array $A[1,\cdots,s]$ of length $s$.
For each element $u\in[m]$, $x(u),y(u),z(u) \in [s]$ are the three distinct locations in $A$ that are
probed to determine whether $u$ is in the set or not. For each set $S\subseteq[m]$ 
of size at most $n$, the assignment $\sigma(S) \in \{0,1\}^s$ to $A$ is
such that for all elements $u\in[m]$, $\Maj(A[x(u)],A[y(u)],A[z(u)])$ is 1 iff $u\in S$, where $\Maj$ 
is the MAJORITY of 3 bits.

\begin{definition}(model-graph for $\Phi$, third vertex, meet)
	Let $\Phi$ be a $(m,n,s,3)$-scheme with MAJORITY as the query function. Fix a graph $G$ such that
$V(G)=[s]$, $|E(G)|=m$ and edge labels: $\{\lab(e)|e\in E(G)\}=[m]$ (there is a unique edge for each label 
in $[m]$). $G$ is called a {\em model-graph} for $\Phi$ if
for each $u\in [m]$ the edge labelled $u$ has the set of endpoints in $\{\{x(u),y(u)\},\{y(u),z(u)\},\{z(u),x(u)\}\}$.
For example, the graph $G=([s], \{x(u)\upseudoedge{u} y(u)| u\in [m]\})$ is a model graph for $\Phi$.

In a model-graph for $\Phi$, let $e$ be the set of endpoints of the edge with label $u$. The element in the singleton 
$(\{x(u),y(u),z(u)\} \setminus e)$ is defined to be the \textit{third} vertex of $u$.

Two edge-disjoint cycles $C_1$ and $C_2$ are said to \textit{meet} in a model-graph for $\Phi$ if there exist
elements $u,v\in[m]$ such that the third vertices of $u$ and $v$ are the same vertex and the edges labelled $u$
and $v$ are in the different cycles $C_1$ and $C_2$ respectively.
\end{definition}

\begin{definition}
A model-graph $G$ for an 
$(m,n,s,3)$-scheme with MAJORITY as the query function is said to be \textit{forced}
if at least one of the following three conditions hold.
\begin{enumerate}
	\item[(P1)] $\exists$ edge-disjoint odd cycles $C_1, C_2$ in $G$ with lengths at most $n$ each that
		intersect at a vertex.

\item[(P2)] $\exists$ edge disjoint even cycles $C_1, C_2$ in $G$ with lengths at most $n$ each and
	$C_1$ and $C_2$ meet.

\item[(P3)] $\exists$ an even cycle $C$ of length at most $n$, such that some two edges in $C$, labelled $e$ and $f$ say,
	have an even number of edges between them (while traversing the edges of the cycle in order) and the third 
	vertices of $e$ and $f$ are the same vertex.
\end{enumerate}
\end{definition}

\begin{lemma} \label{lm:modelGraphIsNeverForced}
A model-graph for a scheme with MAJORITY as the query function cannot be forced.
\end{lemma}

\begin{lemma} \label{lm:schemeImpliesModelGraphForced}
	Any $(m,n,\floor{\frac 1 6m^{1-\frac 1{\floor {\frac n 2} + 1}}},3)$-scheme with MAJORITY as the query
	function has a forced model-graph.
\end{lemma}

From lemmas~\ref{lm:modelGraphIsNeverForced} and \ref{lm:schemeImpliesModelGraphForced}, it follows that when MAJORITY is used as the query function,
$s_N(m,n,3)>\frac 1 6m^{1-\frac{1}{\floor{\frac n 2} +1}}$.

\begin{proof}[Proof of Lemma~\ref{lm:modelGraphIsNeverForced}]
Fix a $(m,n,s,3)$-scheme $\Phi$ with MAJORITY as the query function. Fix a model-graph $G$ for $\Phi$.
Assume $G$ is forced, that is, it satisfies (P1) or (P2) or (P3) above.
\begin{description}
	\item[Case: (P1) holds.] (P1) implies that there are edge-disjoint cycles $C_1$ and $C_2$ in $G$ such that, 
		\[C_1: u_0 \pseudoedge{e_1} u_1 \pseudoedge{e_2} \cdots
		\pseudoedge{e_{2k+1}} u_{2k+1}=u_0,\] \[C_2: u_0 \pseudoedge{f_1} v_1 \pseudoedge{f_2} \cdots
		\pseudoedge{f_{2l+1}} v_{2l+1}=u_0,\] and $2k+1, 2l+1 \leq n$.
		Let $S_0=\{e_1,e_3,\cdots,e_{2k+1}\}\cup\{f_2,f_4,\cdots,f_{2l}\}$ and $S_1=\{e_2,e_4,\cdots,e_{2k}\}\cup\{f_1,f_3,\cdots,f_{2l+1}\}$.
		Note that $|S_0|=|S_1|\leq n$. We claim that $\Phi$ cannot represent any set $S$ such that \[S_0\subseteq S \subseteq \bar{S_1}.\]
	In particular, $\Phi$ cannot represent the set $S_0$. Assume $\Phi$ represents such an $S$. We claim that $u_0$ cannot be assigned a 0.
	If $u_0$ is assigned a 0, then since $e_1\in S$, $u_1$ must be assigned a 1. Otherwise, $\Maj(A[x(e_1)],A[y(e_1)],A[z(e_1)])=\Maj(0,0,b)=0$,
	where $b$ is the bit assigned to the location in $\{x(e_1),y(e_1),z(e_1)\}\setminus \{u_0,u_1\}$. Since, $u_1$ is assigned a 1 and $e_2\notin S$, $u_2$ must be assigned a 0. Similarly, since $e_3\in S$, $u_3$ must be assigned
	a 1 and so on. Finally, $u_{2k+1}=u_0$ must be assigned a 1. A contradiction. Hence $u_0$ cannot be assigned a 0.

	Again, we claim $u_0$ cannot be assigned a 1. For if $u_0$ is assigned a 1, since $f_1\notin S$, $v_1$ must be assigned a 0. Again, since
	$f_2\in S$, $v_2$ must be assigned a 1 and so on. Finally, $v_{2k+1}=u_0$ must be assigned a 0. A contradiction. 

	Since $u_0$ can neither be assigned a 0 or a 1, we get a contradiction.
	
	\begin{remark}In the proofs below, we will often encounter similar arguments, where we will have a cycle of dependencies: assigning
a particular bit to a location will force the assignment to the next location along the cycle.
\end{remark}

\item[Case: (P2) holds.] (P2) implies that there are edge-disjoint cycles $C_1$ and $C_2$ in $G$ such that, 
		\[C_1: u_0 \pseudoedge{e_1} u_1 \pseudoedge{e_2} \cdots
		\pseudoedge{e_{2k}} u_{2k}=u_0,\] \[C_2: v_0 \pseudoedge{f_1} v_1 \pseudoedge{f_2} \cdots
		\pseudoedge{f_{2l}} v_{2l}=v_0,\] $2k, 2l \leq n$, and the third vertices of $e_1$ and $f_1$ are the same vertex $w$. 
		Let $S_0=\{e_1,e_3,\cdots,e_{2k-1}\}\cup\{f_2,f_4,\cdots,f_{2l}\}$ and $S_1=\{e_2,e_4,\cdots,e_{2k}\}\cup\{f_1,f_3,\cdots,f_{2l-1}\}$.
		Note that $|S_0|=|S_1|\leq n$. We claim that $\Phi$ cannot represent any set $S$ such that \[S_0\subseteq S \subseteq \bar{S_1}.\]
	In particular, $\Phi$ cannot represent the set $S_0$. Assume $\Phi$ represents such an $S$. Since $e_1\in S$, either 
	the location $u_0$ or the location $u_1$ of the memory A must be assigned a 1,
	otherwise $\Maj(A[x(e_1)],A[y(e_1)],A[z(e_1)])=\Maj(A[u_0],A[u_1],A[w])=\Maj(0,0,A[w])=0.$ Assume $u_1$ is assigned
	a 1. Then, since $e_2$ is not in the set, using a similar argument, $u_2$ must be assigned a 0. Similarly, $u_3$ must be assigned a 1
	and so on. Finally, $u_{2k}=u_0$ must be assigned a 0. Similarly, if $u_0$ was assigned a 1, then $u_1$ must be assigned a 0. Thus, 
	$\Maj(x(e_1),y(e_1),z(e_1))=\Maj(0,1,A[w])=A[w]$. Hence $w$ must be assigned a 1.

	Again, since $f_1$ is not in $S$, either $v_0$ or $v_1$ is assigned a 0. If $v_1$ is assigned 0,
	$v_2$ must be assigned a 1, $v_3$ a 0, and so on. Finally, $v_{2l}=v_0$ must be assigned a 1. Similarly, if $v_0$ is assigned 1, then $v_1$ is assigned 
	a 0. Therefore, $\Maj(x(f_1),y(f_1),z(f_1))=\Maj(A[v_0],A[v_1],A[w])=\Maj(0,1,A[w])=A[w]$. Since $f_1\notin S$, $w$ 
	must be assigned a 0. A contradiction.

\item[Case: (P3) holds.] (P3) implies that there is a cycle $C$: \[v_0 \pseudoedge{e_1} v_1 \cdots \pseudoedge{e_{2k}} v_{2k} \cdots 
	\pseudoedge{e_{2l}} v_{2l}=v_0,\] $2k\leq2l\leq n$ and the third vertices of $e_1$ and $e_{2k}$ are the same vertex $w$.
	Let $S_0=\{e_1,e_3,\cdots,e_{2l-1}\}$ and $S_1=\{e_2,e_4,\cdots,e_{2k},\cdots,e_{2l}\}$. Note that $|S_0|=|S_1|\leq n$. 
	We claim that $\Phi$ cannot represent any set $S$ such that \[S_0\subseteq S \subseteq \bar{S_1}.\] In particular, $\Phi$
	cannot represent the set $S_0$. Assume $\Phi$ represents such an $S$. Since $e_1\in S$, either $v_0$ or $v_1$ must be assigned
	a 1. Assume that $v_1$ is assigned a 1. Then, since $e_2\notin S$, $v_2$ must be assigned a 0. Again, since $e_3\in S$, $v_3$
	must be assigned a 1 and so on. All locations in $R:=\{v_{2r}| 0\leq r\leq l\}$ must be assigned a 0 and all locations in $Q:=\{v_{2r+1}|
	0\leq r\leq l-1\}$ must be assigned a 1. Else if $v_0$ is assigned a 1, then all locations in $R$ must be assigned a 1 and 
	all locations in $Q$ must be assigned a 0. Now, $\Maj(x(e_1),y(e_1),z(e_1))=\Maj(A[v_0],A[v_1],A[w])=\Maj
	(0,1,A[w])=A[w]$. Since $e_1\in S$, $w$ must be assigned a 0. Similarly, $\Maj(x(e_{2k}),y(e_{2k}),z(e_{2k}))=\Maj
	(A[v_{2k-1}],A[v_{2k}],A[w])=\Maj(0,1,A[w])=A[w]$. Since $e_{2k}\notin S$, $w$ must be assigned a 1. A contradiction.
\end{description}
\end{proof}

In order to prove Lemma~\ref{lm:schemeImpliesModelGraphForced} we will make use of the following proposition, which is a consequence of
a theorem of Alon, Hoory and Linial~\cite{AHL2002} (see also Ajesh Babu and Radhakrishnan~\cite{BR2009}).

\begin{proposition} \label{prop:densityImpliesSmallCycle}\label{prop:Moore}
	Fix a graph $G$ such that the average degree $d\geq 2$. Then, \[(d-1)^k > |V(G)| \implies \exists \text{ a  cycle } C\subseteq E(G), |C| \leq 2k.\]
\end{proposition}

\begin{proof}[Proof of Lemma~\ref{lm:schemeImpliesModelGraphForced}] 
	Fix an $(m,n,\floor{\frac 1 6 m^{1-\frac 1{\floor{\frac n 2} + 1}}},3)$-scheme $\Phi$ 
	that uses MAJORITY as the query function. Note that 
	$s:=\floor{\frac 1 6 m^{1-\frac 1{\floor{\frac n 2} + 1}}}$ implies 
	\[m \geq s^{1 + \frac 1 {\floor{\frac n 2}}} + 4s + 1.\tag*{($\star$)}\] For $\Phi$ we will come up with a model-graph which is forced, that is, one of (P1), (P2)
	or (P3) holds. We will start with 
	an initial model-graph $G$ for $\Phi$. We will observe that the average degree of $G$ is high and invoke 
	Proposition~\ref{prop:densityImpliesSmallCycle} to find a small cycle $C$. If $|C|$ is odd, we will bin it in $\text{ODD}$, delete $C$ and repeat.
	If $|C|$ is even and all the third vertices of the labels of edges in $C$ are distinct, we will bin $C$ in $\text{EVEN}$, delete the edges of $C$ and repeat;
	otherwise, we will either discover that property (P3) holds or we will modify our model-graph and find an odd cycle in it and bin it in $\text{ODD}$,
	delete it and repeat. The moment the sum of the lengths of the deleted cycles exceeds $2s$, we know either the sum of the lengths of odd or even cycles exceeds $s$ and
	two odd cycles intersect or even cycles with distinct third vertices meet, which means either (P1) or (P2) holds. Formally, the procedure can be 
	described as below. We will maintain the following invariant. $\text{EVEN}$ will contain edge-disjoint cycles of length even and at most $n$ each and the
	third vertices of the labels in such a cycle will be all distinct. $\text{ODD}$ will contain edge-disjoint cycles of length odd and at most $n$. Furthermore
	$([s],E(G)\cup\text{EVEN}\cup{ODD})$ will always be a model-graph for $\Phi$.

	\begin{description}
		\item[Step 0: Initialization.]
			$\text{EVEN}=\emptyset$. $\text{ODD}=\emptyset$. $G=([s], \{x(u)\upseudoedge{u} y(u)| u\in [m]\})$. Observe $G$ is
			a model-graph for $\Phi$.
		\item[Step 1.] If $\sum_{C\in\text{EVEN}\cup\text{ODD}}|C| > 2s$ , END (this ensures that either (P1) or (P2) holds). 
			Else, using Proposition~\ref{prop:densityImpliesSmallCycle} fix a cycle $C \subseteq E(G)$ such that $|C|\leq n$.
		\item[Step 2.] If $|C|$ is odd, $\text{ODD} \leftarrow \text{ODD} \cup \{C\}$ and $E(G)\leftarrow E(G)\setminus C$ and GOTO Step 1.
		\item[Step 3.] If $|C|$ is even and all the third vertices of the labels of edges in $C$ are distinct, $\text{EVEN}\leftarrow
			\text{EVEN}\cup\{C\}$ and $E(G)\leftarrow E(G)\setminus C$ and GOTO Step 1.
		\item[Step 4.] If $|C|$ is even and the third vertices of the labels of two edges in $C$ which have an even number of edges between
			them while traversing the edges of $C$ in order, then END (Note this means that (P3) holds).
		\item[Step 5.] If $|C|$ is even and the third vertices of the labels of two edges in $C$ have an odd number of edges between them (while
			traversing the edges of $C$ in order), then represent $C$ as 
		\[C: v_0 \pseudoedge{e_1} v_1 \cdots v_{2k} \pseudoedge{e_{2k+1}} v_{2k+1} \cdots \pseudoedge{e_{2l}} v_{2l}=v_0,\]
		such that the third vertices of $e_1$ and $e_{2k+1}$ are the same vertex $w$. We modify the model-graph $G$ by changing
		the endpoints of the edges appearing with labels $e_1,e_{2k+1}$ from $\{v_0,v_1\}$, $\{v_{2k},v_{2k+1}\}$ to  $\{v_1,w\}$, $\{v_{2k},w\}$ respectively,
		thus obtaining a shorter odd cycle $C'$ in $G$:
		\[E(G)\leftarrow (E(G)\setminus\{v_0\upseudoedge{e_1}v_1, v_{2k}\upseudoedge{e_{2k+1}}v_{2k+1}\})\cup \{v_1\upseudoedge{e_1}w, v_{2k}\upseudoedge{e_{2k+1}}w\}\] 
		(Observe: $G$ with $E(G) \cup \{e| e\in \text{ODD} \cup 
		\text{EVEN}\}$ continues to be a model-graph for $\Phi$).
		\[C'\subseteq E(G): w\pseudoedge{e_1}v_1 \pseudoedge{e_2} v_2 \cdots v_{2k}\pseudoedge{e_{2k+1}}w\] is an odd length cycle of length at most $n$ in 
		$G$. 
		
		$\text{ODD} \leftarrow \text{ODD} \cup \{C'\}$. $E(G)\leftarrow E(G)\setminus C'$. GOTO Step 1.
	\end{description}

	In Step 1, if $|E(G)| \leq 2s$, then the average degree $d$ is at least $\frac {m-2s} s > s^\frac 1{\floor{\frac n 2}}+2$ (from $\star$) and $(d-1)^{\floor{\frac n 2}}>s$
	which implies from Proposition~\ref{prop:densityImpliesSmallCycle} that there is a cycle of length at most $n$.

	We claim that the procedure terminates only by encountering an END statement in Step 1 or in Step 4. 
	Observe that once the procedure finds a cycle in Step 1, then exactly one of the four if conditions in Steps 2-5 holds. If the procedure does not
	encounter an END statement in Step 4, then the procedure moves to Step 1 again as each of the Steps 2, 3 and 5 end in a `GOTO Step 1' statement.

	If the procedure encounters the END statement in Step 4, then (P3) holds. If the procedure encounters the END statement in Step 1, then from
	the pigeonhole principle, either $\sum_{C \in \text{ODD}}|C| > s$ or $\sum_{C \in \text{EVEN}}|C| > s$. In the first case, (P1) holds. In the second case,
	since each edge in a cycle in $\text{EVEN}$ has a distinct third vertex, two cycles in $\text{EVEN}$ meet.

	Finally, we observe that the procedure terminates. If the procedure does not terminate in Step 4, then the procedure repeatedly finds edge disjoint cycles
	and deletes them. If the number of edges in the deleted cycle exceeds $2s$, then the procedure will terminate when it encounters the END statement in 
	Step 1.
\end{proof}

\subsection{Density argument}

The query functions considered in this section all admit a density argument. For such a query function, 
a valid scheme that supports sets from a large universe using
only small space, must conceal a certain dense graph
that avoids certain forbidden configurations. On the other hand, standard graph theoretic results (e.g., the Moore bound) would imply
that dense graphs must have at least one of those forbidden configurations, which would contradict the existence of any such scheme. 

In this section we provide lower bounds for the following ten query functions: AND
function, $(x,y,z)\mapsto (x\oplus y)\wedge z$, $(x,y,z)\mapsto (x\vee
y)\wedge z$, the ALL-EQUAL function, $(x,y,z)\mapsto (x\wedge y\wedge
z)\vee (\bar{y}\wedge\bar{z})$, and their complements. For these functions,
we deal with two functions---a function and its
complement---with a single proof. In these proofs, we produce sets $S$
and $T$ of size at most $n$ such that storing $S$ and not storing $T$
leads to a contradiction. The proof for the complement function works
with a small twist: storing $T$ and not storing $S$ leads to the
contradiction.

We now develop a common framework to prove lower bounds for the above mentioned query functions. We assume the query function is $f$, where $f$ could be any of these ten functions. Fix a scheme $\Phi_f$ for the query function $f$,
where the memory consists of three distinct bit arrays: $A[1,\cdots,s], B[1,\cdots,s]$ and $C[1,\cdots,s]$. 
For any element $u\in[m]$, the scheme $\Phi_f$ probes three distinct locations $x(u)\in A,y(u)\in B,z(u)\in C$
to determine if $u$ is in the set or not. 
Given any set $S\subseteq[m]$ of size at most $n$,
the assignment $\sigma_f(S)\in \{0,1\}^{3s}$ to the memory is such that for all elements $u\in[m]$,
$f(A[x(u)],B[y(u)],C[z(u)])$ is 1 iff $u\in S$.

We will need the following definitions. 
\begin{definition}[$G_{A,B}(E), G_{B,C}(E), G_{A,C}(E)$]
	For a scheme $\Phi_f$, and any subset of elements $E\subseteq[m]$, we define the bipartite graph $G_{A,B}(E)$
	as follows. The vertex sets are $A=[s]$ and $B=[s]$. For each element $u\in E$, we have an edge labelled
	$u$ with end points $x(u)\in A$ and $y(u)\in B$.

	Similarly, we define the bipartite graphs $G_{B,C}(E)=(B:=[s],C:=[s],\{y(u)\upseudoedge{u} z(u) : u\in E\})$
	and $G_{A,C}(E)=(A:=[s],C:=[s],\{x(u)\upseudoedge{u} z(u): u\in E\})$.
\end{definition}

\begin{definition}[private vertex, private $z$ vertex]
	For an element $u\in[m]$, we say $u$ does not have a private vertex if each of its three probe locations is also
	the probe location of another element, that is, there exists elements $v_1,v_2,v_3\in[m]$ such that
	$u\notin \{v_1,v_2,v_3\}$, $x(u)=x(v_1)$, $y(u)=y(v_2)$ and $z(u)=z(v_3)$. If it is not the case 
	that $u$ does not have a private vertex, we say that $u$ has a private vertex.

For an element $u\in[m]$, we say $u$ has a private $z$ vertex if its probe location in $C$ is not shared by any other element, that is for all $v\in[m]\setminus\{u\}$, $z(v)\neq z(u)$.
\end{definition}

\begin{proposition}\label{prop:privateVertex} 
	If every $u\in[m]$ has a private vertex, then $3s\geq m$.
\end{proposition}

\begin{proposition}\label{prop:privateZVertex}
	If at least $\frac m 2$ $u$'s in $[m]$ have a private $z$ vertex, then $s\geq \frac m 2$.
\end{proposition}


\subsubsection{AND}

Let the query function $f$ be the AND function.

From Proposition~\ref{prop:privateVertex}, if every $u\in[m]$ has a private vertex, we are done. Otherwise, fix an element $u\in[m]$ such that $u$ 
does not have a private vertex. Let elements $v_1,v_2,v_3\in[m]$ be such that $x(u)=x(v_1)$, $y(u)=y(v_2)$
and $z(u)=z(v_3)$. Let $S=\{v_1,v_2,v_3\}$ and $T=\{u\}$. Clearly, $|S|,|T|\leq n$. We now show that $\Phi_f$ cannot 
represent any set $S'$ such that
\[ 
	S\subseteq S' \subseteq [m]\setminus T.
\]
In particular $\Phi_f$ cannot represent the set $S$. To show this, consider the assignment for the set $S'$:
$\sigma_f(S')$ to the memory $A,B,C$. Since $v_1,v_2,v_3\in S'$, and $f$ is the AND function, $A[x(v_1)]=
B[y(v_2)]=C[z(v_3)]=1$. This implies $A[x(u)]=B[y(u)]=C[z(u)]=1$. Thus, $f(A[x(u)],B[y(u)],C[z(u)])=1$. But
$u\notin S'$. A contradiction.

The same argument works for $\bar{f}$ when we try to represent any set $S'$ such that
\[
	T\subseteq S' \subseteq [m]\setminus S.
\]

\subsubsection{$(x,y,z)\mapsto (x\oplus y)\wedge z$}

Let the query function $f$ be $(x,y,z)\mapsto (x\oplus y)\wedge z$.

From Proposition~\ref{prop:privateZVertex}, if at least $\frac m 2$ elements have a $z$ private vertex, we are done. Otherwise, fix a set $E\subseteq [m]$
of size at least $\frac m 2$ that has no element with a private $z$ vertex.

Now, we assume $s\leq \frac 1 4m^{1-\frac{1}{\lfloor\frac n 2\rfloor}}$ and prove a contradiction. 
The average degree of vertices in $G_{A,B}(E)$ is at least $2m^{\frac{1}{\lfloor\frac n 2\rfloor}}\geq 2$, for large $m$.
Since
\[
	(2m^{\frac{1}{\lfloor\frac n 2\rfloor}}-1)^{\lfloor\frac n 2\rfloor}\geq (m^{\frac{1}{\lfloor\frac n 2\rfloor}})^
{\lfloor\frac n 2\rfloor}\geq m> 2s,
\] from Proposition~\ref{prop:Moore}, 
there exists a cycle of length $2k\leq 2\lfloor\frac n 2\rfloor\leq n$ in $G_{A,B}(E)$:
\[
 v_0 \pseudoedge{u_1} v_1 \pseudoedge{u_2} \cdots
		\pseudoedge{u_{2k}} v_{2k}=v_0.
\]

Since $u_1\in E$ does not have a private $z$ vertex, fix an element $v$ such that $z(v)=z(u_1)$. Let
$S=\{v,u_2,u_3,\cdots,u_{2k}\}$ and $T=\{u_1\}$. Clearly, $|S|,|T|\leq n$. We now show that $\Phi_f$ cannot 
represent any set $S'$ such that
\[ 
	S\subseteq S' \subseteq [m]\setminus T.
\]
In particular $\Phi_f$ cannot represent the set $S$. Under the assignment $\sigma_{S'}$, since $v\in S'$, 
the location $z(v)=z(u_1)$ is assigned 1; otherwise if $C[z(v)]=0$, then $(A[x(v)]\oplus B[y(v)])\wedge C[z(v)]=0$.
Again, under the assignment $\sigma_{S'}$, exactly one of $v_0$ and $v_1$ is assigned 0. To see this, let us assume $v_1$ 
is assigned the
bit $b$, then since $u_2\in S'$, $v_2$ must be assigned the bit $\bar{b}$; otherwise if both $v_1$ and $v_2$ are assigned
$b$, then $(A[x(u_2)]\oplus B[y(u_2)])\wedge C[z(u_2)]=(b\oplus b)\wedge C[z(u_2)]=0$. Similarly, $v_3$ must be assigned the
bit $b$, $v_4$ must be assigned the bit $\bar{b}$ and so on. Thus, each location in $\{v_1,v_3,\cdots,v_{2k-1}\}$ must be assigned
the bit $b$ and each location in $\{v_2,v_4,\cdots,v_{2k}=v_0\}$ must be assigned
the bit $\bar{b}$. Now,  $(A[x(u_1)]\oplus B[y(u_1)])\wedge C[z(u_1)]=(b\oplus\bar{b})\wedge 1=1$. But $u_1\notin S'$.
A contradiction. Thus, $s>\frac 1 4m^{1-\frac{1}{\lfloor\frac n 2\rfloor}}$.

The same argument works for $\bar{f}$ when we try to represent any set $S'$ such that
\[
	T\subseteq S' \subseteq [m]\setminus S.
\]
\subsubsection{$(x,y,z)\mapsto (x\vee y)\wedge z$}

Let the query function $f$ be $(x,y,z)\mapsto (x\vee y)\wedge z$.

From Proposition~\ref{prop:privateZVertex}, if at least $\frac m 2$ elements have a private $z$ vertex, we are done. Otherwise, fix a set $E\subseteq [m]$
of size at least $\frac m 2$ that has no element with a private $z$ vertex. Now make disjoint pairs of distinct elements $(u,v)$,  $u,v\in E$, 
that share the same $z$ location, that is, $z(u)=z(v)$. Make as many pairs as possible. At most $s$ many, one per 
location in $C$, elements can remain unpaired. Delete these unpaired elements from $E$. 
Thus, $m':=|E|\geq\frac 1 2 (\frac m 2 -s)\geq \frac m 8$ if $s\leq \frac m 4$ (which if not true, we are immediately done). For each such pair $(u,v)$, define $v$ to be the reserved partner of $u$. Let
$E'\subseteq E$ be the set of all unreserved elements. Clearly, $E'=\frac{|E|}{2}=\frac {m'}{2}\geq\frac m{16}$.

Find elements $u,v,w\in E'$ such that $u\notin\{v,w\}$, $x(u)=x(v)$ and $y(u)=y(w)$. If no such triple of elements exists then each element
has at least one of its $x$ or $y$ probe locations not shared with any other element in $E'$, and thus $2s\geq |E'|\geq \frac m{16}$ and
we are done.

Let $S=\{u\}\cup\{$reserved partner of $t: t\in\{v,w\}\}$, $T=\{v,w\}$. Clearly, $|S|,|T|\leq n$.
We now show that $\Phi_f$ cannot 
represent any set $S'$ such that
\[ 
	S\subseteq S' \subseteq [m]\setminus T.
\]
In particular $\Phi_f$ cannot represent the set $S$. Any assignment $\sigma_{S'}$ must assign 1 to each location $z(t)$,
for each $t\in\{v,w\}$; since, the reserved partner of each such $t$ is in $S'$ and $t$ shares the same $z$ vertex with it. 
The assignment $\sigma_{S'}$ must assign 0 to each of $x(v),y(v),x(w),y(w)$; since, if any of them is assigned a 1, together with the fact that $z(v),z(w)$ are assigned 1, either $(A[x(v)]\vee B[y(v)])\wedge C[z(v)]$ or
$(A[x(w)]\vee B[y(w)])\wedge C[z(w)]$ will evaluate to 1 which cannot happen as $v,w\notin S'$.

Now, since $x(u)=x(v)$ and $y(u)=y(w)$, $(A[x(u)]\vee B[y(u)])\wedge C[z(u)]=(0\vee 0)\wedge C[z(u)]=0$. But $u\in S'$. 
A contradiction.

The same argument works for $\bar{f}$ when we try to represent any set $S'$ such that
\[
	T\subseteq S' \subseteq [m]\setminus S.
\]
\subsubsection{ALL-EQUAL}

Let the query function $f$ be the ALL-EQUAL function.

We assume $s< \frac{1}{600}m^{1-\frac{1}{\lfloor\frac n 4\rfloor +1}}$ and derive a contradiction.

\begin{lemma}\label{lem:forbiddenConfig}
	If $G_{A,B}$ and $G_{B,C}$ contain cycles
	$C_1$ and $C_2$ respectively of size at most $\frac n 2$ each such that for some $e_1\in C_1$ and $e_2\in C_2$, 
	labels of $e_1$ and $e_2$ are the same, then the scheme $\Phi_f$ cannot represent all sets of size $n$.
\end{lemma}

\begin{lemma}\label{lem:densityImpliesForbiddenConfig}
	If $s< \frac{1}{600}m^{1-\frac{1}{\lfloor\frac n 4\rfloor+1}}$, then $G_{A,B}$ and $G_{B,C}$ will contain cycles
	$C_1$ and $C_2$ respectively each of length at most $\frac n 2$ such that for some $e_1\in C_1$ and $e_2\in C_2$, 
	labels of $e_1$ and $e_2$ are the same.
\end{lemma}

Our claim follows immediately from the above two lemmas.

\begin{proof}[Proof of Lemma~\ref{lem:forbiddenConfig}]
	Let the two cycles $C_1$ and $C_2$ as promised by Lemma~\ref{lem:forbiddenConfig} be
\begin{enumerate}
	\item[]\[ C_1: r_0 \pseudoedge{u_1} r_1 \pseudoedge{u_2} \cdots
		\pseudoedge{u_{2k}} r_{2k}=r_0,\]
	\item[] \[C_2: r_0 \pseudoedge{u_1} t_1 \pseudoedge{v_2} \cdots
		\pseudoedge{v_{2l}} t_{2l}=r_0,\]
\end{enumerate}
where $r_0=y(u_1)\in B$, $r_1=x(u_1)\in A$, $t_1=z(u_1)\in C$ and $2k,2l\leq \frac n 2$.
Let $S=\{u_2,u_3,\cdots,u_{2k}\}\cup\{v_2,v_3,\cdots,v_{2l}\}$ and $T=\{u_1\}$. Clearly, $|S|,|T|\leq n$.
We now show that $\Phi_f$ cannot represent any set $S'$ such that
\[ 
	S\subseteq S' \subseteq [m]\setminus T.
\]
In particular $\Phi_f$ cannot represent the set $S$. Let the assignment $\sigma_{S'}$ assign
the bit $b$ to the location $r_0$. Then, all locations in $\{r_1,r_2,\cdots,r_{2k}\}\cup\{t_1,t_2,\cdots,t_{2l}\}$ must be
assigned the bit $b$: since, $r_{2k}=r_0$ is assigned the bit $b$ and since $u_{2k}\in S'$, $r_{2k-1}$ must be assigned the
bit $b$, similarly since $r_{2k-1}$ is assigned the bit $b$ and since $u_{2k-1}\in S'$, $r_{2k-2}$ must be assigned the
bit $b$ and so on and thus all locations in $\{r_1,r_2,\cdots,r_{2k}\}$ must be assigned the bit $b$. Arguing similarly,
since $r_0=t_{2l}$ is assigned the bit $b$ and since $v_{2l}\in S'$, $t_{2l-1}$ must be assigned the bit $b$, and so on
and thus all locations in $\{t_1,t_2,\cdots,t_{2l}\}$ must be assigned the bit $b$. Now, since $r_0,r_1,t_1$ are all assigned
the bit $b$ and they are the three probe locations for the element $u_1$, the ALL-EQUAL function for $u_1$ evaluates to 1. But
$u_1\notin S'$. A contradiction.
\end{proof}

\begin{proof}[Proof of Lemma~\ref{lem:densityImpliesForbiddenConfig}]
	In the graph $G_{A,B}([m])$, as long as the number of edges is at least
	$6s^{1+{\lfloor\frac n4\rfloor}}$ we can find a cycle
	of length at most $\frac n2$ from Proposition~\ref{prop:Moore}; since, the average degree then would be at least 
	\[\frac{6s^{1+{\lfloor\frac n4\rfloor}}}{2s}\geq3s^{\frac 1{\lfloor\frac n4\rfloor}}\geq 2,\]
	and since \[(3s^{\frac 1{\lfloor\frac n4\rfloor}}-1)^{\lfloor\frac n4\rfloor}>2s.\] 

	We repeatedly remove cycles: $C_1, C_2,\cdots$ from the graph $G_{A,B}$  
	each of length at most $\frac n2$ using Proposition~\ref{prop:Moore} (that is, we delete the edges appearing in the picked cycle),
	till no more cycle remains; then, the number of remaining edges in the graph will be at most 
	$6s^{1+{\lfloor\frac n4\rfloor}}\leq\frac 1{100}m$; since,
	\[
		s< \frac{1}{600}m^{1-\frac{1}{\lfloor\frac n 4\rfloor+1}}\implies m>600s^{1+\frac 1{\lfloor\frac n4\rfloor}}.
	\]

	Similarly, following the same argument, we can remove cycles: $D_1, D_2,\cdots$ from the graph $G_{B,C}([m])$, where
	the length of each cycle is at most $\frac n2$, till no more cycle remains. The number of remaining edges in the graph
	will be at most $\frac1{100}m$.

	Now, pick a random element $u\in[m]$. The probability that the edge $\{x(u),y(u)\}$ with label $u$ 
	appears in some cycle $C_i$ and the edge $\{y(u),z(u)\}$ with label $u$ appears in some cycle $D_j$ is
	at least $1-(\frac 1{100}+\frac 1{100})=\frac {98}{100}$ (using the union bound).
	Thus, there exists an element $u$ and cycles $C_i$ and $D_j$ each of length at most $\frac n 2$ in $G_{A,B}$ and $G_{B,C}$ respectively, each containing an edge labelled $u$.
\end{proof}

The same argument works for $\bar{f}$ when in Lemma~\ref{lem:forbiddenConfig}, we try to represent any set $S'$ such that
\[
	T\subseteq S' \subseteq [m]\setminus S.
\]

\subsubsection{$(x,y,z)\mapsto (x\wedge y\wedge z)\vee (\bar{y}\wedge\bar{z})$}

Let the query function $f$ be $(x,y,z)\mapsto (x\wedge y\wedge z)\vee (\bar{y}\wedge\bar{z})$.

\begin{definition}[forced]
	We say that the scheme $\Phi_f$ is forced if in the graph $G_{B,C}([m])$ at least
	one of the following two conditions hold.
	\begin{enumerate}
		\item[(P1)] There exists a cycle $C$ of length at most $\frac n 2$ such that
			there are two elements $u_1,u_2\in[m]$ which appear as labels of some
			two edges in $C$ and $x(u_1)=x(u_2)$.
		\item[(P2)] There exist cycles $C_1$ and $C_2$ of lengths at most $\frac n 2$ each
			and some two elements $u,v\in[m]$ that appear as the labels of an edge in $C_1$
			and an edge in $C_2$ respectively have the same $x$ probe location, that is, $x(u)=x(v)$.
	\end{enumerate}
\end{definition}

\begin{lemma}\label{lem:forcedImpliesContradiction}
	If the scheme $\Phi_f$ is forced, then it cannot represent all sets of size at most $n$.
\end{lemma}

\begin{lemma}\label{lem:densityForces}
	If $s< \frac{1}{7}m^{1-\frac{1}{\lfloor\frac n 4\rfloor+1}}$, then $\Phi_f$ is forced.
\end{lemma}

Our claim follows immediately from these two lemmas.

\begin{proof}[Proof of Lemma~\ref{lem:forcedImpliesContradiction}]
We assume either (P1) or (P2) above holds and derive a contradiction.
\begin{description}
	\item[Case: (P1) holds.]
		Let $C$ be a cycle:
		\[ C:  r_0 \pseudoedge{v_1} r_1 \pseudoedge{v_2} \cdots r_{k-1} 
			\pseudoedge{v_k} r_{k}\cdots\pseudoedge{v_l} r_l=r_0
		\]
		in $G_{B,C}$ of length at most $\frac n2$, that
		contains two edges labelled $v_1$ and $v_k$ such that $x(v_1)=x(v_k)$.
		Let $S=\{v_2,\cdots,v_l\}$ and $T=\{v_1\}$. Clearly $|S|,|T|\leq n$
		We now show that $\Phi_f$ cannot represent any set $S'$ such that
\[ 
	S\subseteq S' \subseteq [m]\setminus T.
\]
In particular $\Phi_f$ cannot represent the set $S$. The assignment $\sigma_{S'}$ has to 
assign the same bit to all locations in $\{r_1,\cdots,r_l=r_0\}$. To see this, first observe that for any
element $u\in S'$, for the query function $f$ to evaluate to 1, $y(u)$ and $z(u)$ 
must be assigned the same bit. Now, if $r_1$ is assigned the bit $b$,
then, since $v_2\in S$, $r_2$ must be assigned the bit $b$, and again $r_2$ is assigned the bit $b$ 
and $v_3\in S$, $r_3$ must be assigned the bit $b$ and so on and thus all locations in $\{r_1,\cdots,r_l=r_0\}$ is
assigned the same bit $b$. Now, since the locations $x(v_1),x(v_k)$ are the same and the locations $y(v_1),z(v_1), y(v_k),z(v_k)$ are assigned the same bit the query function will evaluate to the same value for both $v_1$ and $v_k$. But $v_k\in S'$ and $v_1\notin S'$. A contradiction.
\item[Case: (P2) holds.]
	Let $C_1$ and $C_2$ be cycles:
	\begin{align*}
	&C_1:  r_0 \pseudoedge{u_1} r_1 \pseudoedge{u_2} \cdots r_{k-1} 
			\pseudoedge{u_k} r_{k}\cdots\pseudoedge{u_l} r_l=r_0\\
	&C_2:  t_0 \pseudoedge{v_1} t_1 \pseudoedge{v_2} \cdots\pseudoedge{v_p} t_p=t_0
 \end{align*}
	each of length at most $\frac n 2$ and $x(u_k)=x(v_1)$.
	Let $S=\{u_2,u_3\cdots u_l\}\cup\{v_2,v_3,\cdots,v_p\}$ and $T=\{u_1,v_1\}$. Clearly, $|S|,|T|\leq n$.
	We now show that $\Phi_f$ cannot represent any set $S'$ such that
	\[ 
		S\subseteq S' \subseteq [m]\setminus T.
	\]
In particular $\Phi_f$ cannot represent the set $S$. Following a similar reasoning as in the earlier case, we have the
assignment $\sigma_{S'}$ must assign all locations in $\{r_1,r_2,\cdots, r_l=r_0\}$ the same bit $b$ and
all locations in $\{t_1,t_2,\cdots,t_p\}$ the same bit $b'$. Since $u_1\notin S'$, 
the locations in $\{y(u_1),z(u_1)\}=\{r_0,r_1\}$ cannot be assigned both 0's and since both of them are assigned
the same bit $b$, it must be that $b=1$. Now, since $u_k\in S'$ and $\{y(u_k),z(u_k)\}=\{r_{k-1},r_k\}$ are 
both assigned $b=1$, $x(u_k)$ must be assigned 1. $x(u_k)=x(v_1)$ is assigned 1 and the locations $y(v_1), z(v_1)$ are
assigned the bit $b'$ implies $f(A[x(v_1)],B[y(v_1)],C[z(v_1)])=f(1,b',b')=1$. But $v_1\notin S'$. A contradiction.
\end{description}
\end{proof}

\begin{proof}[Proof of Lemma~\ref{lem:densityForces}]
If the scheme uses space less and (P1) above does not hold then we show (P2) above must hold.

In the graph $G_{B,C}([m])$, as long as the number of edges is at least
$6s^{1+{\lfloor\frac n4\rfloor}}$ we can find a cycle
of length at most $\frac n2$ from Proposition~\ref{prop:Moore}; since, the average degree then would be at least 
\[\frac{6s^{1+{\lfloor\frac n4\rfloor}}}{2s}\geq3s^{\frac 1{\lfloor\frac n4\rfloor}}\geq 2,\]
and since \[(3s^{\frac 1{\lfloor\frac n4\rfloor}}-1)^{\lfloor\frac n4\rfloor}>2s.\] 

We repeatedly remove cycles: $C_1, C_2,\cdots$ from the graph $G_{B,C}$  
each of length at most $\frac n2$ using Proposition~\ref{prop:Moore} till no more cycle remains; then, the number of remaining edges in the graph will be at most 
$6s^{1+{\lfloor\frac n4\rfloor}}$; since,
\[
	s< \frac{1}{7}m^{1-\frac{1}{\lfloor\frac n 4\rfloor+1}}\implies m>7s^{1+\frac 1{\lfloor\frac n4\rfloor}}.
\]
Thus, the sum of the lengths of the removed cycles exceeds $s$. Now, for any two edges labelled $u$, $v$ in a removed cycle,
$z(u)\neq z(v)$; otherwise (P1) holds. Thus, the sum of the probe locations in $C$ ($z$ probe locations) of the labels of the edges appearing in the cycles
exceeds $s$ and hence some two cycles must have edge labels $u$, $v$ respectively, such that $z(u)=z(v)$.
\end{proof}

The same argument works for $\bar{f}$ when in Lemma~\ref{lem:forcedImpliesContradiction}, we try to represent any set $S'$ such that
\[
	T\subseteq S' \subseteq [m]\setminus S.
\]


\subsection{Dimension argument}

In the proofs below we make use of algebraic arguments.
The query functions considered in this section admit a dimension argument. The argument will go as follows. To each element in the
universe we will associate a vector. If the space used by a scheme is small, then all these vectors will reside in a vector space of 
small dimension, which, in turn, will force a linear dependence between them. We will then argue that if we keep one element $u$ aside
and not store any other elements appearing in the linear dependence, then the scheme will be left with no choice for $u$ leading to a contradiction.

We now develop a common framework to prove lower bounds for the query functions- PARITY and $(x,y,z)\mapsto 1$ iff $x+y+z\neq1$ (over $\mathbb R$).
We assume the query function is $f$, where $f$ could be either of these two functions. Fix a scheme $\Phi_f$ for the query function $f$,
where the memory is a bit array $A[1,\cdots,s]$ consisting of $s$ locations.
For any element $u\in[m]$, the scheme $\Phi_f$ probes three distinct locations $x(u),y(u),z(u)\in A$
to determine if $u$ is in the set or not. 
Given any set $S\subseteq[m]$ of size at most $n$
the assignment $\sigma_f(S)\in \{0,1\}^{s}$ to $A$ is such that for all elements $u\in[m]$,
$f(x(u),y(u),z(u))$ is 1 iff $u\in S$.

\newcommand{\vect}{\mathsf{vector}} 

\begin{definition}[Fields $\mathbb F_2$, $\mathbb R$, vector spaces $V$, $W$, $\vect$]

	Let $\mathbb F_2=\{0,1\}$ be the field with modulo 2 arithmetic. Let $\mathbb R$ be the field of real numbers with the usual arithmetic.

	Let $V={\mathbb F}^s_2$ be the vector space of all $s$ dimensional vectors over the field $\mathbb F_2$.

	Let $W={\mathbb R}^s$ be the vector space of all $s$ dimensional vectors over the field $\mathbb R$.

	For an element $u\in[m]$, define $\vect(u)$ to be an $s$-dimensional vector which contains three 1's in the positions $x(u),y(u)$ and $z(u)$
	and 0's everywhere else. Clearly $\vect(u)\in V\subseteq W$.
\end{definition}

\subsubsection{PARITY}
Now we fix $f$ to be the parity function. 

\begin{proposition}\label{prop:dotAssignment}
	For a set $S\subseteq [m]$ of size at most $n$, the assignment $\sigma_f(S)$ is such that for all $u\in [m]$, the dot product of the two vectors $\vect(u)$ and $\sigma_f(S)$ in $V$ is 1 iff $u\in S$. 
\end{proposition}

If $m>s$, then the set of $m$ vectors $\{\vect(u):u\in[m]\}\subseteq V$ is linearly dependent; $V$ is $s$-dimensional.
Thus, there exists distinct elements $u,v_1,\cdots,v_t\in [m]$, such that
\[
	\vect(u)=\sum_{i=1}^t \vect(v_i).
\]
Taking dot product on both sides with the assignment vector $\sigma_f(\{u\})$ for the singleton $\{u\}$, we have
\begin{align*}
	\vect(u).\sigma_f(\{u\})&=\sum_{i=1}^t \vect(v_i).\sigma_f(\{u\})\\
	\implies 1&=\sum_{i=1}^t 0=0,
	\end{align*}
	where the last step follows from Proposition~\ref{prop:dotAssignment} and the fact that  $u\in\{u\}$ and $v_i\notin\{u\}$ for all $i\in[t]$.
	A contradiction. Thus $s\geq m$.
	
\subsubsection{$(x,y,z)\mapsto1$ iff $x+y+z\neq 1$}
Now we fix $f$ to be the $(x,y,z)\mapsto1$ iff $x+y+z\neq 1$ function. 

\begin{proposition}\label{prop:dotAssignmentForR}
	For a set $S\subseteq [m]$ of size at most $n$, the assignment $\sigma_f(S)$ is such that for all $u\in [m]$, the dot product of the two vectors $\vect(u)$ and $\sigma_f(S)$ in $W$ is 0 iff $u\in S$. 
\end{proposition}

If $m>s$, then the set of $m$ vectors $\{\vect(u):u\in[m]\}\subseteq W$ is linearly dependent; $W$ is $s$-dimensional.
Thus, there exists distinct elements $u,v_1,\cdots,v_t\in [m]$ and $\alpha_1,\cdots,\alpha_t\in \mathbb R$, such that
\begin{align}
	\vect(u)&=\sum_{i=1}^t\alpha_i \vect(v_i).\label{eqn:sum}
\end{align}
Taking dot product on both sides with the assignment vector $\sigma_f(\{u\})$ for the singleton $\{u\}$, we have
\begin{align*}
	\vect(u).\sigma_f(\{u\})&=\sum_{i=1}^t\alpha_i \vect(v_i).\sigma_f(\{u\})\\
	\implies 0&=\sum_{i=1}^t \alpha_i
	\end{align*}
	where the last step follows from Proposition~\ref{prop:dotAssignmentForR} and the fact that  $u\in\{u\}$ and $v_i\notin\{u\}$ for all $i\in[t]$.

Again, taking dot product on both sides of Equation (\ref{eqn:sum}) with the assignment $\sigma_f(\emptyset)$ for the empty set, we have
\begin{align*}
	\vect(u).\sigma_f(\{u\})&=\sum_{i=1}^t\alpha_i \vect(v_i).\sigma_f(\{u\})\\
	\implies 1&=\sum_{i=1}^t \alpha_i=0
	\end{align*}
	where the last step follows from Proposition~\ref{prop:dotAssignmentForR} and the fact that  $u,v_i\notin\emptyset$ for all $i\in[t]$.
	A contradiction. Thus, $s\geq m$.

\subsection{Degree argument}

In this section we provide lower bound proofs for the query functions $(x,y,z)\mapsto (x\wedge y)\oplus z$
and $(x,y,z)\mapsto 1$ iff $x+y+z=1$. 

\subsubsection{$(x,y,z)\mapsto (x\wedge y)\oplus z$}\label{sec:xy+zLbForAllN}
Let $\Phi$ be a scheme with $(x,y,z)\mapsto (x\wedge y)\oplus z$ as the query function. 
The memory consists of three distinct bit arrays: $A[1,\cdots,s], B[1,\cdots,s]$ and $C[1,\cdots,s]$. 
For any element $u\in[m]$, the scheme probes three locations $x(u)\in A$, $y(u)\in B$ and $z(u)\in C$
to determine if $u$ is in the set or not. We treat each location as a boolean variable.
Given any set $S\subseteq[m]$ of size at most $n$
the assignment $\sigma(S)\in \{0,1\}^{3s}$ to $A,B$ and $C$ is such that for all elements $u\in[m]$,
$(x(u)\wedge y(u))\oplus z(u)$ is 1 iff $u\in S$.

We first prove that $s=\Omega(\sqrt{mn})$ by specializing the lower bound proof in~\cite{RSV2002} 
to our case.

\begin{definition}[Field $\mathbb F_2$, vector space $V$, polynomials $P_S$]

	Let $\mathbb F_2$ denote the field $\{0,1\}$ with mod 2 arithmetic. The query function
$(x,y,z)\mapsto (x\wedge y)\oplus z$ is same as $(x,y,z)\mapsto xy+z$ (over $\mathbb F_2$).

Let $V$ be the vector space over the field $\mathbb F_2$ of all multilinear polynomials
of total degree at most $2n$ in the $3s$ variables: $A[1],\cdots,A[s],$
$B[1],\cdots,B[s],$ $C[1],\cdots,C[s]$ with coefficients coming from $\mathbb F_2$.

For each set $S\subseteq[m]$, we define the polynomial $P_S$ in $3s$ variables and coefficients
coming from the field $\mathbb{F}_2$  as follows:
\[
	P_S=\prod_{u\in S} (x(u)y(u)+z(u)).
\]

We make $P_S$ multilinear by reducing the exponents of each variable using the identity $x^2=x$ for each
variable $x$. This identity holds since we will be considering only $0$-$1$ assignment to the variables.

\end{definition}

To prove the theorem for $(x,y,z)\mapsto (x \wedge y)\oplus z$, we use the following two lemmas.

\begin{lemma}\label{lm:linearIndependence}
	The set of $m \choose n$ multilinear polynomials $\{P_S: |S|=n\}$ is linearly independent in the vector space $V$.
\end{lemma}

\begin{lemma}\label{lm:spanningSet}
	$V$ has a spanning set of size at most ${{3s+2n}\choose {2n}}$.
\end{lemma}

Using these two lemmas, we first prove the theorem and provide the proofs of the lemmas later.

\begin{proof}[Proof]
Now, since the size of a linearly independent set is at most the size of a spanning set, using
Lemmas~\ref{lm:linearIndependence} and \ref{lm:spanningSet}, we have
\begin{align*}
	{m \choose n} &\le {{3s+2n}\choose {2n}} \\
	\implies \left({\frac m n}\right)^n &\le \left({\frac {e(3s+2n)}{2n}}\right)^{2n}\\
	\implies 3s &\ge \frac 2 e\sqrt n (\sqrt m - e\sqrt n)\\
	\implies 3s &\ge \frac{18}{10e} \sqrt{mn} \quad \text{(when $n\le \frac m{900}\{\implies e\sqrt n \le \frac 1{10}\sqrt m\}$).}
\end{align*}	

When $n\geq \frac m{900}$, the fact that the assignments to the memory for storing different sets 
of size $\lceil \frac m{900}\rceil$ are different implies that the space required is at least 
$\lg{m \choose {\lceil\frac m{900}}\rceil}\ge \Omega(m)\ge\Omega(\sqrt{mn})$.

\end{proof}

Now, we prove the two lemmas.
\begin{proof}[Proof of Lemma~\ref{lm:linearIndependence}]

First observe that any $S$ of size $n$, the polynomial $P_S$ has $n$ factors of degree 2 each.
	Hence, the degree of $P_S$ is at most $2n$.
		
	For sets $S,S'\subseteq[m]$ of size $n$ each, the evaluation of the polynomial
$P_S$ on the assignment $\sigma(S')$ is
\[
	P_{S}(\sigma(S')) = \begin{cases}
	\hfill	0 \hfill & \text{ if $S\neq S'$} \\
	\hfill	1 \hfill & \text{ if $S = S'$.} \\
	\end{cases}
\]
Since $S\neq S'$ and $|S|=|S'|=n\geq 1$, there exists $u\in S$ such that $u\notin S'$ and thus under the
assignment $\sigma(S')$, the factor $x(u)y(u)+z(u)$ in $P_S(\sigma(S'))$ evaluates to 0. While,
when $S=S'$, for each $u\in S$ the factor $x(u)y(u)+z(u)$ in $P_S(\sigma(S'))$ evaluates to 1.

In particular, this proves that $\{P_S:|S|=n\}$ has size $m\choose n$. 
Further we use this observation below to prove the lemma.
	
	Let $\sum_{S:|S|=n}\alpha_S P_S = 0$ where each $\alpha_S\in \mathbb F_2$.
	To show that the $P_S$'s are linearly independent, we need to show that each $\alpha_S$ is 0.
	Consider an arbitrary set $S'$ of size $n$, consider the assignment $\sigma(S')$
	to the variables in the above identity.
\begin{align*}	
0 &= \sum_{S:|S|=n}\alpha_S P_S(\sigma(S')) \\
  &= \alpha_{S'} P_{S'}(\sigma(S'))\ \ + \sum_{S:S\neq S',|S|=n}\alpha_S P_S(\sigma(S'))\\
&= \alpha_{S'} P_{S'}(\sigma(S')) \quad   \text{(since, $P_S(\sigma(S'))=0$ for each $S\neq S'$)}\\
&=	\alpha_{S'}   		  \quad   \text{(since, $P_{S'}(\sigma(S'))=1$)}.
	\end{align*}
\end{proof}

\begin{proof}[Proof of Lemma~\ref{lm:spanningSet}]
	The monomials of total degree at most $2n$ form a spanning set; each polynomial in $V$ can
	be written as a linear combination of these monomials.
	Thus, the size of this spanning set is
	   \begin{align*}
		    \sum_{k=0}^{2n} {{3s}\choose{k}}
		    \le &{{3s+2n}\choose{2n}},
	    \end{align*}
	    where the last inequality follows from the fact that $T\mapsto T\cap[3s]$ is an onto
	    map from ${[3s+2n]\choose{2n}}$ to ${[3s]}\choose{\le 2n}$.
    \end{proof}
    \subsubsection{$(x,y,z)\mapsto 1$ iff $x+y+z=1$}

The lower bound proof for $(x,y,z)\mapsto 1$ iff $x+y+z=1$ is similar to the lower bound proof for $(x,y,z)\mapsto (x\wedge y)\oplus z$.
The only difference here is that instead of looking at the query function over the field $\mathbb F_2$, we consider the query function over the field $\mathbb F_3$
(the set of three elements $\{0,1,2\}$ with mod 3 arithmetic). Over the field $\mathbb F_3$, the query function $(x,y,z)=1$ iff $x+y+z=1$
is same as $(x,y,z)\mapsto x+y+z+xy+yz+zx$ (a degree 2 polynomial). Accordingly, the multilinear polynomial corresponding to a set $S$ of size $n$ is defined to be
\[
	P_S=\prod_{u\in S} (x(u)+y(u)+z(u)+x(u)y(u)+y(u)z(u)+z(u)x(u)).
\]
where we reduce the exponents using the identity $x^2=x$ for each variable $x$ (this identity holds as we consider only $0$-$1$ assignments).
Notice that $P_S$ has degree at most $2n$ and the rest of the proof is same as before.

\subsection{Improved bound for $(x,y,z)\mapsto (x\wedge y)\oplus z$}
We now prove part (c) of Theorem~\ref{result3}. We extend the idea of the lower bound proof for the query function $(x,y,z)\mapsto (x\wedge y)\oplus z$ for 
    $\lg m \leq n \leq \frac m{\lg m}$ to get a better lower bound. We continue to use the framework of Section~\ref{sec:xy+zLbForAllN}.

For the scheme $\Phi$, we first define a bipartite graph $G_{A,B}$ with vertex sets $A,B\subseteq[s]$. 
The sets $A, B$ and
the edges will be determined as follows. Initially let $A=B=[s]$ and for each element 
$u\in[m]$ add an edge labelled $u$ between $x(u)\in A$ and $y(u)\in B$. 
Our goal is to ensure that there are $m'\geq \frac{4m}{5}$ edges in the graph and
all vertices have degree at least $\frac {m}{10s}\geq \frac{m'}{10s}$. 
Repeatedly delete all vertices with degree less than $\frac{m}{10s}$, to get $A$ and $B$.
We will lose at most $2s\times
\frac{m}{10s}=\frac m 5$ edges. Also delete the corresponding elements of the universe $[m]$ which appeared as labels of the deleted edges. $\Phi$
restricted to the remaining elements of the universe gives a scheme on  $m'\geq \frac {4m}5$ elements. Without loss of generality, we assume the remaining $m'$
elements form the set $[m']$. Proving
a lower bound of $\Omega(\sqrt{{m'}n \frac{\lg \frac {m'} n}{\lg \lg {m'}}})=\Omega(\sqrt{mn\frac{\lg\frac m n}{\lg \lg m}})$ on the space for this restricted scheme will prove the theorem. 

Thus, $G_{A,B}$ has $m'$ edges labelled with distinct elements from the universe $[m']\subseteq[m]$ and vertex sets $A,B\subseteq[s]$ with minimum degree at least $\frac{m'}{10s}$. The edge labelled $u\in[m']$ has endpoints $x(u)\in A$ and $y(u)\in B$.

We assume $s=\sqrt{cm'n}$ for $c=\left\lceil \frac{\lg\frac {m'}{16n}}{\lg\lg\frac{m'}{n}}\right\rceil-1$ and show a contradiction.

\begin{definition}[The parameters $s,m,m',n,c,c',D,k$ and their relations]
	In the following, we assume 
	\begin{align*}
		&c=\left\lceil\frac{\lg\frac {m'}{16n}}{\lg\lg\frac{m'}{n}}\right\rceil-1,\\
		&\lg m\leq n \leq \frac {m}{\lg m},\\
		&s=\sqrt{cm'n}.
	\end{align*}
	From before, we have	
	\begin{align*}
		&m\geq m'\geq \frac 4 5 m.
	\end{align*}
	Define 
	\begin{align*}
		&D:=\frac 1{10}\sqrt{\frac{m'}{cn}},\\
		&c':=2e\cdot e^2\cdot(64\cdot100^2)^2c\ln c,\\
		&k:=\left\lceil\frac{n}{2}(1-\frac{1}{c'})\right\rceil.
	\end{align*}

Clearly, using the above definitions we have 
minimum degree of a vertex in $G_{A,B}$ is at least $D$ (by the assumption on $s$), 
\begin{align*}
	&m'\geq \max\{2.100^2cn, 64\cdot e^2n, 2^{12^2e^2}n\},\\
	&n\geq\max\{4c',6c,64\cdot100^2c\},\\ 
	&\min\{D,c,c',n\}\geq100
\end{align*}
for all large $m'$. These relations will be used in proving the various lemmas and the theorem
below.
\end{definition}

\newcommand{\restr}{\mathsf{restriction}} 

\begin{definition}[Trap, gain, good gainer, restriction]
	In a bipartite graph $(A,B,E)$, two edges $e_1$ and $e_2$ are said to trap an edge $e=\{a,b\}$ if 
	$a\in e_1$ and $b\in e_2$.

	In a bipartite graph $(A,B,E)$, a set $S\subseteq E$ is said to gain if
	there exists two edges $e_1, e_2 \in S$ which either intersect at an end point, that is, $e_1 \cap e_2\neq \emptyset$
	or if they trap an edge $e \notin S$. 

	In a bipartite graph $(A,B,E)$, a set of edges $S\subseteq E$ of size $n$ is called a good gainer if
	every subset $S'\subseteq S$ of size $2k$ gains ($k$ is defined above).

	For a set $T\subseteq[m]$, define $\restr(T)$ to be the set of all assignments to the memory 
	specified by the storage scheme for sets which do not contain any element from $T$, 
	that is, $\restr(T)=\{\sigma(S):|S|\leq n \text{ and } T\cap S = \emptyset\}$.
\end{definition}

\begin{lemma}\label{gainImpliesGoodGainer}
	Fix a bipartite graph $G=(A,B,E)$ with  $|A|,|B|\leq s$, $|E|=m'$ and minimum degree at least $D$. A uniformly random set of edges $S\subseteq E$ of size $n$ is a good gainer with probability at least
	$\frac 1 2$ for all large $m'$.
\end{lemma}

\begin{lemma}\label{goodGainerImpliesLowDegree}
	If $S\subseteq [m']$ such that $E_S:=\{e\in E(G_{A,B}):\lab(e)\in S\}$ is a good gainer in $G_{A,B}$, then there exist
	a polynomial $\hat{P_S}$ of degree at most $2n-\lfloor\frac {n} {2c}\rfloor$ and a set $T_S\subseteq[m']$ such that $T_S\cap S=\emptyset$, $|T_S|
	\leq \frac n {2c}$ and 
	$P_S(\bar{x})=\hat{P_S}(\bar{x})$ for every $\bar{x}\in\restr(T_S)$.
\end{lemma}

Using these two lemmas, we first prove the theorem, and provide the proofs of the lemmas later.

\begin{proof}[Proof of part (c)]

To each set $S\subseteq[m']$ of size $n$ such that $E_S:=\{e\in E(G_{A,B}):\lab(e)\in S\}$ is a good gainer, we associate a
polynomial $\hat{P_S}$ of degree at most $2n-\lfloor\frac n {2c}\rfloor$ and set $T_S\subseteq [m']$ such that $|T_S|\leq \frac n{2c}$,
$T_S \cap S = \emptyset$ and
$P_S(\bar{x})=\hat{P_S}(\bar{x})$ for every $\bar{x}\in\restr(T_S)$. This is possible due to Lemma~\ref{goodGainerImpliesLowDegree}.

Let $T$ be a random subset of $[m']$ where each element of $[m']$ is independently included in $T$ with probability $\frac 1 2$. For each $S$
above, let $I_S$ be the indicator random variable of the event $S\cap T=\emptyset$ and $T_S\subseteq T$. Then, the expected sum of these
indicator random variables is

\[
	\sum_{S: E_S\text{ is good gainer}}\mathbb{E} [I_S]  
	\geq \frac 1 2 {m' \choose n}\frac 1 {2^n}\cdot\frac 1 {2^{\frac n c}}
	\geq \frac 1{2^{2n}} {m' \choose n},
\]
where in the first inequality, the factor $\frac 1 2 {m' \choose n}$  appears because there are at least that many good gainers (from
Lemma~\ref{gainImpliesGoodGainer}), the factor $\frac 1 {2^n}\cdot\frac 1 {2^{\frac n c}}$ appears because for a given good gainer $S$,
the $n$ elements in $S$ must fall outside $T$ and the elements in $T_S$ (at most $\frac n c$ of them) must fall inside $T$. The second
inequality follows from the fact that $cn\geq n + c$ which holds for $n,c \geq 2$.

Thus, there exists a $T$ for which at least  $\frac 1{2^{2n}} {m' \choose n}$ indicator random variables $I_S$'s are 1. 
Fix such
a $T$ and let $\mathcal{S}=\{S \subseteq[m']: S\text{ is a good gainer, } S\cap T=\emptyset \text{ , } T_S\subseteq T\}$. 
We have $|\mathcal{S}|\geq \frac 1{2^{2n}} {m' \choose n}$.

We first prove that the collection of polynomials $\{\hat{P_S}|S\in \mathcal{S}\}$ is linearly independent. 
First, observe that $R\subseteq T$ implies that $\restr(T)\subseteq \restr(R)$. 
Thus, for any $S\in\mathcal{S}$, $\hat{P_S}(\bar{x})=P_S(\bar{x})$ for all $\bar{x}\in\restr(T)\subseteq\restr(T_S)$. 
Since, each $S\in\mathcal{S}$ is disjoint
from $T$, the assignment $\sigma(S)\in\restr(T)$. Consequently, for all $S,S'\in\mathcal{S}$, we have the property that
\[
	\hat{P_{S}}(\sigma(S'))=
	P_{S}(\sigma(S')) = \begin{cases}
	\hfill	0 \hfill & \text{ if $S\neq S'$} \\
	\hfill	1 \hfill & \text{ if $S = S'$,} \\
	\end{cases}
\]
which can be used to prove $\{\hat{P_S}|S\in \mathcal{S}\}$ is linearly independent as we did in Lemma~\ref{lm:linearIndependence}.

Since the polynomials $\hat{P_S}$ for $S\in\mathcal{S}$ are linearly independent and have degree at most $2n-\lfloor \frac n {2c} \rfloor
\leq 2n-\frac n{3c}=2n(1-\frac 1{6c})$ (for $n\geq 6c$),
we have 
\begin{align}
	\frac 1 {2^{2n}} {m' \choose n}^n &\leq {3s + 2n-\lfloor \frac n {2c}\rfloor  \choose 2n-\lfloor \frac n {2c}\rfloor} \nonumber \\
	\implies \frac 1 {2^{2n}} \left(\frac {m'} n\right)^n &\leq \left(e\cdot\frac {3s + 2n}{2n(1 -\frac 1{4c})}\right)^{2n(1-\frac 1{6c})\nonumber}\\
	\implies 2n\left(1 -\frac 1{4c}\right)\left(\frac 1 2 \sqrt{\frac {m'} n}\right)^{\frac 1 {1 - \frac 1 {6c}}}
	&\leq 3es + 2en\nonumber\\
	\implies n\left(\frac 1 2 \sqrt{\frac {m'} n}\right)^{1+\frac 1{6c}} &\leq 3es +2en\label{firstStep}\\
	\implies n\left(\frac 1 4 \sqrt{\frac {m'}n}\right)^{1+\frac 1{6c}} &\leq 3e\sqrt{cm'n}\label{secondStep}\\
	\implies \left(\frac 1 {16}\cdot \frac {m'} n\right)^{\frac 1 {12 c}} &\leq 12 e \sqrt{c}\nonumber\\
	\implies c&\geq\frac {\lg\frac {m'}{16n}}{12\lg\lg\frac {m'}{n}}\label{lastStep},
\end{align}
where (\ref{firstStep}) holds because $c\geq \frac{1}{2}\implies 1-\frac{1}{4c}\geq \frac{1}{2}$ and $1\geq 1-(\frac{1}{6c})^2\implies 
\frac{1}{1-\frac{1}{6c}}\geq1+\frac{1}{6c}$, (\ref{secondStep}) holds for $s=\sqrt{cm'n}$ and $m'\geq 64e^2n\implies \frac1 2\sqrt{\frac{m'}{n}}-2e\geq
\frac 1 4\sqrt{\frac{m'}{n}}$, 
and (\ref{lastStep}) holds because 
$c<\frac {\lg\frac {m'}{16n}}{12\lg\lg\frac {m'}{n}}\implies \left(\frac 1 {16}\cdot \frac {m'} n\right)^{\frac 1 {12 c}}> \lg\frac{m'}{n}$ and
$m'\geq 2^{12^2e^2}n\implies\sqrt{\lg\frac{m'}{n}}\geq 12e\implies\lg{\frac{m'}{n}}\geq
12e\sqrt{\lg\frac{m'}{n}}\geq12e\sqrt c$. A contradiction.

\end{proof}

\begin{proof}[Proof of Lemma~\ref{goodGainerImpliesLowDegree}]
	For a $S\subseteq[m']$ such that $E_S$ is a good gainer, 
	we get the desired $T_S$ and $\hat{P_S}$ by running the following procedure. Recall
	$P_S=\prod_{u\in S} (x(u)y(u)+z(u))$ consists of $n$ factors each of degree $2$. 
	
	Initially we set $\hat{S} = S$, $E_{\hat{S}} = E_S$, $\hat{P_S} = P_S$, and  $T_S=\emptyset$. 
	We run the procedure for $\lfloor \frac n {2c}\rfloor$ steps. In each step we maintain the following
	invariants: $T_S\cap S=\emptyset$, $P_S(\bar{x})=\hat{P_S}(\bar{x})$ for all assignments $\bar{x}\in \restr(T_S)$.
	At the end of each step we delete two elements from each of $\hat{S}$ and $E_{\hat{S}}$. 	
	\begin{description}
	\item[IF] $E_{\hat{S}}$ has two intersecting edges $e_1$ and $e_2$ (with labels $v$ and $w$ say), that is, $x(v)=x(w)$
	or $y(v)=y(w)$,
	then we multiply out the two factors in $\hat{P_S}$ corresponding to $v$ and $w$ to get a degree 3 factor:
	\begin{align*} 
		&(x(v)y(v)+z(v))(x(w)y(w)+z(w))\\
		=&x(v)x(w)y(v)y(w)+x(v)y(v)z(w)+x(w)y(w)z(v)\\
		=&\begin{cases}
			\hfill x(v)y(v)y(w)+x(v)y(v)z(w)+x(w)y(w)z(v)\hfill & \text{if $x(v)=x(w)$}\\
			\hfill x(v)x(w)y(v)+x(v)y(v)z(w)+x(w)y(w)z(v)\hfill & \text{if $y(v)=y(w).$}
		\end{cases}
	\end{align*}

\item[ELSE] (no two edges in $E_{\hat{S}}$ intersect) since we run the procedure for only $\lfloor \frac {n}{2c}\rfloor$
	steps and delete some 2 elements at the end of each step, 
	$|\hat S|\geq n-2(\lfloor \frac n {2c} \rfloor -1) \geq 2k=2\lceil \frac n 2(1-\frac 1 c)\rceil$. 
	We invoke Lemma~\ref{gainImpliesGoodGainer} with $S'=E_{\hat{S}}$, 
	to find edges $e_1$ and $e_2$ (with labels $v$ and $w$ say) in $E_{\hat{S}}\subseteq E_S$ 
	that traps edge $e\notin E_S$ (with label $t$ say). 
	Note that $e$ is not just not in $E_{\hat{S}}$ as promised by Lemma~\ref{gainImpliesGoodGainer} but also not in $E_S$;
	otherwise $e$ and $e_1$ would be both in $E_S$ and intersecting which will satisfy the IF part. We add the element
	$t$ to the set $T_S$. Observe that for any assignment in $\restr(T_S)$, $x(t)y(t)+z(t)=0$, that is, $x(t)y(t)=z(t)$. 
	We use this relation, to simplify the product of the two factors in $\hat{P_S}$ corresponding to $v$ and $w$ 
	to get a degree 3 factor:
	\begin{align*}
	&(x(v)y(v)+z(v))(x(w)y(w)+z(w))\\
	=&x(v)x(w)y(v)y(w)+x(v)y(v)z(w)+x(w)y(w)z(v)\\
	=& 
	\begin{cases}
		\hfill	 z(t)x(w)y(v)+x(v)y(v)z(w)+x(w)y(w)z(v)\hfill & \text{if trapped edge $e=\{x(v),y(w)\}$} \\
		\hfill	 z(t)x(v)y(w)+x(v)y(v)z(w)+x(w)y(w)z(v)\hfill & \text{if trapped edge $e=\{x(w),y(v)\}$.} \\
	\end{cases}
	\end{align*}
	\item We delete $v$ and $w$ from $\hat{S}$ and delete
		$e_1$ and $e_2$ from $E_{\hat{S}}$ and repeat.
	\end{description}

	At the end of the procedure, that is, after $\lfloor\frac {n}{2c}\rfloor$ steps,
	clearly $|T_S|\leq\lfloor\frac {n}{2c}\rfloor$, $S\cap T_S=\emptyset$,
	degree of $\hat{P_S}\leq 2n-\lfloor\frac n {2c}\rfloor$, and  $\hat{P_S}(\bar{x})=P_S(\bar{x})$ 
	for all assignments $\bar{x}\in \restr(T_S)$.
\end{proof}

To prove Lemma~\ref{gainImpliesGoodGainer} we require the following lemma, which we prove later.

\begin{lemma}\label{densityImpliesGain}
	Fix any bipartite graph $G= (A,B,E)$ with $|A|,|B|\leq s$, $|E|=m'$ and minimum degree at least $D$. A uniformly random set of edges $S'\subseteq E$ of size $2k$ gains with probability at least
	$1 - \exp(-\frac n {32\cdot 100^2 c})$ for all large $m'$.
\end{lemma}
\begin{proof}[Proof of Lemma~\ref{gainImpliesGoodGainer}]
	For a uniformly random set $S$ of $n$ edges from $E$, the probability that $S$ is a good gainer can
	be lower bounded using the union bound and Lemma~\ref{densityImpliesGain} as follows.
	\begin{align*}
		\Pr[S\text{ is a good gainer}] 
		&= 1 - \Pr[\exists S'\subseteq S, |S'|=2k, S'\text{ does not gain}]\\
		&\geq 1 - \sum_{S'\subseteq S, |S'|=2k}\Pr[S'\text{ does not gain}]\\
		&\geq 1 - {n \choose 2k}\exp\left(-\frac n {32\cdot 100^2 c}\right)\\
		&\geq 1 - {n \choose n-2k}\exp\left(-\frac n {32\cdot 100^2 c}\right)\\
		&\geq 1 - \left(\frac {en}{n-2k}\right)^{n-2k}\exp\left(-\frac n {32\cdot 100^2 c}\right).
	\end{align*}
	Now, since $2k\geq n(1-\frac 1{c'})$ therefore $n-2k \leq \frac n{c'}$, and since $2k \leq  n(1-\frac 1{c'}) + 2$
	therefore $n-2k \geq \frac n{c'} - 2 \geq \frac n{2c'}$ for $n\geq4c'$. Therefore, using these bounds we get,
	\begin{align*}
		\Pr[S\text{ is a good gainer}] 
		&\geq 1 - \left(\frac {en}{\frac n{2c'}}\right)^{\frac n{c'}}\exp\left(-\frac n {32\cdot 100^2 c}\right)\\
		&= 1 - \exp\left(n\left(\frac{\ln2ec'}{c'} - \frac 1{32\cdot 100^2 c}\right)\right)\\
		&\geq 1 - \exp\left(-\frac n{64\cdot 100^2 c}\right)\\
		&\geq \frac 1 2,
	\end{align*}
	where the last inequality holds for $n\geq 64\cdot 100^2 c$ and the second last inequality holds for 
	$c'=2e\cdot e^2\cdot(64\cdot100^2)^2c\ln c$, 
	since 
	\begin{align*}
		2e\cdot e^2\cdot 64\cdot 100^2 &\geq \ln(2e\cdot e^2\cdot 64\cdot 100^2)^2\\
	\implies 2e\cdot e^2\cdot 64\cdot 100^2 - 2 &\geq \ln(2e \cdot 64\cdot 100^2)(2e\cdot e^2 \cdot 64\cdot 100^2),
\end{align*}
and since $\ln c \geq 1$, we have
\begin{align*}
	 (2e\cdot e^2\cdot 64\cdot 100^2 - 2)\ln c + 2\ln c 
	&\geq \ln(2e \cdot 64\cdot 100^2) (2e\cdot e^2 \cdot 64\cdot 100^2) +\ln c + \ln\ln c\\
	\implies(2e\cdot e^2\cdot 64\cdot 100^2)\ln c 
	&\geq \ln((2e\cdot  64\cdot 100^2)(2e\cdot e^2 \cdot 64\cdot 100^2)c\ln c)\\
	&=\ln 2e c'\\
	\implies \frac {c'}{\ln2ec'}&\geq 64\cdot 100^2 c\\
	\implies \frac 1{32\cdot 100^2c}&\geq 2\frac {\ln 2ec'}{c'}.
	\end{align*}
\end{proof}

\begin{proof}[Proof of Lemma~\ref{densityImpliesGain}]
To aid our calculations, we pick a uniformly random set $S'$ with $2k$ edges in two steps. In step I, we randomly pick $k$ vertices: $v_1,\cdots, v_k$ 
from $A$ \underline{with replacement} with each vertex being picked with probability proportional to its degree,
that is, for all $v\in A$, $\Pr[v_i=v]=\frac {\deg(v)}{m'}$. For each $v_i$ we pick an edge $e_i$ incident on $v_i$
uniformly at random, that is, with probability $\frac 1 {\deg(v_i)}$. Thus for each edge $e\in E$ whose endpoint in $A$ is $v$,
\[
	\Pr[e_i=e]=\Pr[v_i=v]\Pr[e_i=e|v_i=v]=\frac{\deg (v_i)}{m'} \cdot \frac 1{\deg(v_i)}=\frac 1 {m'}.
\]  This way, some $t$ distinct edges 
are picked, where $t\leq k$ (note that $t$ can be strictly less than $k$ if for some $i\neq j$, $e_i=e_j$). 
In step II, we delete the edges picked in step I, and 
pick a set of $2k-t$ edges uniformly at random \underline{without replacement} from the remaining edges.

Let BAD denote the event that $S'$ does not gain.

In step I, after sampling the vertices $\{v_1,\cdots,v_k\}$  we get some fixed multiset 
of $k$ vertices with $t$ distinct vertices. It is enough to show that conditioning on each choice for $v_1,\cdots,v_k$, the conditional probability of BAD is at most $\exp\left(-\frac n {32\cdot 100^2 c}\right)$ for large $m'$. After conditioning, the only source of randomness in the calculations below comes from sampling the $2k$ edges in steps I and II.

To upper bound the probability that 
BAD occurs we have the following two cases. Let $B_0\subseteq B$ 
be the neighbourhood of $\{v_1,\cdots,v_k\}$.

\begin{description}
	\item[Case:]$|B_0|\geq \frac{KD}{100}$. For the event BAD to occur, none of
		the edges picked in step II should be incident on any vertex in $B_0$, 
		the neighbourhood of $\{v_1,\cdots,v_k\}$. For
		if some edge $f=\{a,b\}$ picked in step II is incident on a vertex $b$ in the neighbourhood of some $v_i$,
		then either $f$ and $e_i$ trap an edge not in $S$ with endpoints $b$ and $v_i$ or $f$ intersects
		some edge $e_j$ with endpoints $b$ and $v_i$. Since each
		vertex has degree at least $D$, the number of edges incident on $B_0$ is at least $\frac{kD^2}{100}$
		and after deleting the edges $\{e_1,\cdots,e_k\}$ picked in step I,
		at least $k(\frac{D^2}{100}-1)$ edges remain. Since, at
		least $k$ edges are picked in step II out of at most $m'$ edges, 
		the probability that BAD occurs can be upper bounded as follows.
		\[
			\Pr[\text{BAD}]
		\leq  \left(1 - \frac {k (\frac{D^2}{100}-1)}{m'}\right)^k
		\leq \exp\left(-\frac {k^2(D^2-100)}{100m'}\right).
		\]
		Now, using the fact that $2k\geq n(1-\frac 1 {c'})$, $D\geq\frac 1{10}\sqrt{\frac {m'}{cn}}$, we get
		\[
			\Pr[\text{BAD}]\leq \exp\left(- \frac{n(1-\frac 1{c'})^2}{4\cdot 100^2c}(1-\frac{100^2cn}{m'})\right)
			\leq \exp\left(- \frac{n(1-\frac 1{c'})^2}{8\cdot 100^2c}\right),
		\]
		where the last inequality holds since $m'\geq2\cdot100^2cn$.

	\item[Case:]$|B_0|<\frac {kD}{100}$. For each $i\in[k]$, define $\Gamma'(v_i):=\Gamma(\{v_1,\cdots,v_{i-1}\})\cap 
		\Gamma(v_i)$ and $\deg'(v_i):=|\Gamma'(v_i)|-1$, where $\Gamma$ stands for the neighbourhood function. For the event BAD to occur, the following event BAD$(v_i)$
		must occur for each $v_i$. 
		\begin{description}
			\item[When $v_i \notin\{v_1,\cdots,v_{i-1}\}$,] BAD$(v_i)\equiv e_i$ is not incident on any vertex in 
				$\Gamma'(v_i)$.
			\item[When $v_i=v_j$ for some $j<i$,] BAD$(v_i)\equiv e_i=e_j$, where $j$ is the smallest
				number such that $v_i=v_j$.
		\end{description}
		For if some BAD$(v_i)$ does not occur, then $S'$ gains; in the case when 
		$v_i\notin\{v_1,\cdots,v_{i-1}\}$ and $e_i$ is incident on some vertex $b$ in 
		$\Gamma'(\{v_i\})$, then $b$ must be the neighbour of some vertex $v_j$ where $j<i$ 
		and then either $e_i$ and $e_j$ trap an edge not in $S$ with endpoints $b$, $v_j$
		or $e_i$ intersects $e_j$ at $b$.
		Again, in the case when $v_i=v_j$ for
		some $j<i$ and $e_i\neq e_j$ (where $j$ is the smallest number for which $v_i=v_j$,
		then clearly $e_i$ and $e_j$ intersect at the vertex $v_i$.
		Thus we bound the probability of BAD as follows. 
		\[
			\Pr[\text{BAD}]
			\leq \Pr[\forall i\in[k] \ \text{BAD}(v_i)]
			= \prod_{i=1}^{k}\Pr[\text{BAD}(v_i)],
		\]
		where the equality follows from the fact that, in the conditional space, the BAD($v_i$)'s are independent events.
		Now, if $v_i\notin\{v_1,\cdots,v_{i-1}\}$, 
		\[
			\Pr[\text{BAD}(v_i)]
			\leq 1-\frac{\deg'(v_i)+1}{\deg(v_i)}\leq 1-\frac{\deg'(v_i)}{\deg(v_i)}.
		\]
		
		Else, when $v_i=v_j$ for some $j<i$,
		\[
			\Pr[\text{Bad}(v_i)]
			\leq 1-\frac{\deg'(v_i)}{\deg(v_i)}.
		\]
		Using this upper bound, we have
		\begin{align*}
			\Pr[\text{BAD}]\leq \prod_{i=1}^{k}\Pr[\text{BAD}(v_i)]
			&\leq \prod_{i=1}^k\left(1-\frac{\deg'(v_i)}{\deg(v_i)}\right)
			\leq \exp \left(-\sum_{i=1}^k\frac{\deg'(v_i)}{\deg(v_i)}\right)\\
			&\leq \exp(-\frac{49}{50}k) \leq \exp\left(-\frac{49}{100}n(1-\frac 1 {c'})\right),
		\end{align*}
	where the second last inequality can be shown as follows. \newline
	Since $B_0=\mathop{\dot{\bigcup}}_{i=1}^k \Gamma(v_i)\setminus \Gamma'(v_i)$, we have
		\begin{align*}
			\sum_{i=1}^k \left(\deg(v_i) - (\deg'(v_i)+1)\right) & \leq \frac{kD}{100}\\
			\implies \sum_{i=1}^k \deg(v_i)\left(1-\frac{\deg'(v_i)}{\deg(v_i)}\right) & \leq \frac{kD}{100}+k\\
			\implies D\sum_{i=1}^k \left(1-\frac{\deg'(v_i)}{\deg(v_i)}\right) & \leq \frac{kD}{100} +k\\
			\implies \sum_{i=1}^k \left(1-\frac{\deg'(v_i)}{\deg(v_i)}\right) & \leq \frac{k}{100} +\frac k D\\
			\implies \sum_{i=1}^k -\frac{\deg'(v_i)}{\deg(v_i)} & \geq k - (\frac {k}{100}+\frac k D)\\
			& \geq \frac {49}{50}k,
		\end{align*}
		where the last inequality follows from the fact that $D\geq 100$ for large $m'$.
\end{description}

Thus combining the two cases, for large $m'$, the probability that a random set of $2k$ edges gains 
is at least 
\begin{align*}
1 - \max\{\exp\left(- \frac{n(1-\frac 1{c'})^2}{8\cdot 100^2c}\right), \exp\left(-\frac{49}{100}n(1-\frac 1 {c'})\right)\}\\
\geq 1-\exp\left(- \frac{n(1-\frac 1{c'})^2}{8\cdot 100^2c}\right)
\geq
1-\exp\left(-\frac n {32\cdot 100^2 c}\right),
\end{align*}
where the last two inequalities follow from the fact that $c',c\geq2$.

\end{proof}



\bibliography{non-adaptive}

\end{document}